\documentclass[journal]{IEEEtran}
\usepackage{amsmath}
\usepackage{graphicx}
\usepackage{epsfig}
\usepackage{latexsym}
\usepackage{amsfonts}
\usepackage{amssymb}
\usepackage{psfrag}
\usepackage{cite}
\usepackage{subfigure}
\begin{document}

\title{Optimal Power and Range Adaptation for Green Broadcasting
\thanks{S. Luo and T. J. Lim are with the Department of
Electrical and Computer Engineering, National University of
Singapore (e-mail:\{shixin.luo, eleltj\}@nus.edu.sg).} \thanks{R.
Zhang is with the Department of Electrical and Computer Engineering,
National University of Singapore (e-mail:elezhang@nus.edu.sg). He is
also with the Institute for Infocomm Research, A*STAR, Singapore.}}

\author{Shixin Luo,~\IEEEmembership{Student Member,~IEEE}, Rui Zhang,~\IEEEmembership{Member,~IEEE}, and Teng Joon Lim,~\IEEEmembership{Senior Member,~IEEE}}

\maketitle
\begin{abstract}
Improving energy efficiency is key to network providers maintaining profit levels and an acceptable carbon footprint in the face of rapidly increasing data traffic in cellular networks in the coming years. The energy-saving concept studied in this paper is the adaptation of a base station's (BS's) transmit power levels and coverage area according to channel conditions and traffic load. Cell coverage is usually pre-designed based on the estimated static (e.g. peak) traffic load. However, traffic load in cellular networks exhibits significant fluctuations in both space and time, which can be exploited, through cell range adaptation, for energy saving. In this paper, we design short- and long-term BS power control (STPC and LTPC respectively) policies for the OFDMA-based downlink of a single-cell system, where bandwidth is dynamically and equally shared among a random number of mobile users (MUs). STPC is a function of all MUs' channel gains that maintains the required user-level quality of service (QoS), while LTPC (including BS on-off control) is a function of traffic density that minimizes the long-term energy consumption at the BS under a minimum throughput constraint. We first develop a power scaling law that relates the (short-term) average transmit power at BS with the given cell range and MU density. Based on this result, we derive the optimal (long-term) transmit adaptation policy by considering a joint range adaptation and LTPC problem. By identifying the fact that energy saving at BS essentially comes from two major energy saving mechanisms (ESMs), i.e. range adaptation and BS on-off power control, we propose low-complexity suboptimal schemes with various combinations of the two ESMs to investigate their impacts on system energy consumption. It is shown that when the network throughput is low, BS on-off power control is the most effective ESM, while when the network throughput is higher, range adaptation becomes more effective.
\end{abstract}


\IEEEpeerreviewmaketitle
\setlength{\baselineskip}{1.0\baselineskip}
\newtheorem{definition}{\underline{Definition}}[section]
\newtheorem{fact}{Fact}
\newtheorem{assumption}{Assumption}
\newtheorem{theorem}{\underline{Theorem}}[section]
\newtheorem{lemma}{\underline{Lemma}}[section]
\newtheorem{corollary}{Corollary}
\newtheorem{proposition}{\underline{Proposition}}[section]
\newtheorem{example}{\underline{Example}}[section]
\newtheorem{remark}{\underline{Remark}}[section]
\newtheorem{algorithm}{\underline{Algorithm}}[section]
\newcommand{\mv}[1]{\mbox{\boldmath{$ #1 $}}}

\section{Introduction}\label{sec:introduction}
Mobile data traffic is anticipated to grow many-fold between 2010 and 2020, inducing many technical challenges such as how to improve energy efficiency in order to limit growth in energy consumption to a factor smaller than that of data traffic growth. The drive to make cellular networks more ``green'' starts with base stations (BSs), since they make up a large proportion of the total energy consumed in any cellular network \cite{Fettweis09}.

Cell planning, i.e. placement of BSs and coverage area of each one, is usually based on estimated static (e.g. peak) traffic load. Current research in cellular network planning mainly focus on the practical deployment algorithm design. For example, in \cite{Hanly02}, the authors used stochastic geometry to analyze the optimal macro/micro BS density for energy-efficient heterogeneous cellular networks with QoS constraints. The energy efficiency of heterogeneous networks and the effects of cell size on cell energy efficiency were investigated in \cite{Wang10} by introducing a new concept called area energy efficiency. However, traffic load in cellular networks fluctuates substantially over both space and time due to mobility and traffic burstiness. Therefore, there will always be some cells under light load, and others under heavy load, which suggests that static cell planning based on peak load will not be optimal. Load balancing schemes have thus been proposed in both academia and industry \cite{Son09,Hanly95,Ritte03}, which react to load variations across time and cells by adaptively re-allocating users to cells. In \cite{Son09}, a network-wide utility maximization problem was considered to jointly optimize partial frequency reuse and load-balancing in a multicell network. In \cite{Hanly95,Ritte03}, the authors proposed the ``cell breathing'' technique, which shrinks (or expands) the coverage of congested (or under-loaded) cells by reducing (or raising) the power level, so that the load becomes more balanced.

BSs consume a significant amount of energy (up to $60\%$ of the total network energy consumption \cite{3GppTR}) due to their operational units, e.g., processing circuits, air conditioner, besides radio transmission. Therefore, selectively letting some BSs be switched off according to traffic load can yield substantial energy saving. There have been a few BS on-off switching schemes introduced in the literature. For example, energy saving as a function of the daily traffic pattern, i.e the traffic intensity as a function of time, was derived in \cite{Marsan09}, where it is shown through simulations that energy saving on the order of $25-30\%$ is possible. Centralized and distributed BS reconfiguration algorithms were proposed in \cite{Brunner10}, with simulations showing that the centralized algorithm outperforms the distributed one at the cost of increased complexity and overhead. In \cite{Jardosh09}, the authors considered a wireless local area network (WLAN) consisting of a high density of access points (APs). The resource on-demand (RoD) strategy was introduced to power on or off WLAN APs dynamically, based on the volume and location of user demand.

When some BSs are switched off, their coverage areas need to be served by the remaining active BSs in the network. Such a self-organized network (SON) has been introduced in 3GPP LTE \cite{SON09}. A similar but more flexible method called ``Cell Zooming'' was proposed in \cite{Niu10}, which adaptively adjusts the cell size according to traffic load, user requirements, and channel conditions, in order to balance the traffic load in the network and thereby reduce energy consumption. Energy-efficient cellular network planning with consideration of BSs' ability of cell zooming, which is characterized as cell zooming ratio, was investigated in \cite{Niu11}. However, to the best of our knowledge, a scheme that adapts both coverage range and transmit power (including the possibility of turning off the BS) to minimize the total energy consumed has not been studied in the literature, even under the simple one-cell setup. {This motivates our work, which studies the extreme case of one single-cell system in order to obtain useful insights that could be applied in a general multi-cell environment.}

In this paper, we consider the downlink transmission in an orthogonal frequency-division multiple access (OFDMA) based cellular network. Unlike traditional cellular networks using fixed time and/or bandwidth allocation, we consider that the available time-frequency transmission blocks are dynamically and equally allocated to a random number of active mobile users (MUs). Moreover, the BS is assumed to have two levels of power control: short-term power control (STPC) and long-term power control (LTPC), which correspond to the inherent difference in the time scales of the MUs' average channel gain variations (in e.g. seconds) and traffic density variations (in e.g. hours). STPC sets the transmit power based on each MU's distance from the BS to meet each MU's outage probability requirement over fading, while LTPC (including BS on-off control)\footnote{Note that turning off BS is considered in the LTPC of this paper. Since we focus on the extreme case of a one-cell system in this paper, we assume that any uncovered spatial holes left by the single cell of our interest are to be filled by the surrounding active cells, which cause no interference to the considered cell by a proper frequency assignment scheme.} is implemented according to traffic density variations such that the long-term energy consumption at the BS is minimized under a certain system-level throughput constraint. Under the above broadcast channel setup, a new power scaling law, which relates the (short-term) average transmit power of BS with the given cell coverage range and traffic density, is derived for the case of homogeneous Poisson point process (HPPP) distributed MU locations. Based on the derived power scaling law, we determine the optimal long-term cell adaptation policy by considering a joint range adaptation and LTPC problem. Since it is challenging to obtain closed-form expressions for the optimal policy, approximate solutions are derived in closed-form under a high spectrum efficiency (HSE) assumption, which provide further insights into the design of cellular networks of the future in which both power and spectral efficiency are important.

By identifying that the energy saving at BS essentially comes from two major energy saving mechanisms (ESMs), i.e. range adaptation and BS on-off power control, we propose low-complexity suboptimal schemes with various combinations of the two ESMs to further investigate their effects on the system energy consumption. By numerical simulations, it is shown that significant energy saving can be achieved in OFDMA broadcast channels with the optimal cell adaptation policy. Furthermore, it is revealed that when the network throughput requirement is modest, the simple BS on-off control is nearly optimal for cell adaptation in terms of energy saving; however, when higher network throughput is required, a finer-grained strategy of range adaptation is needed. These results provide useful guidelines for designing energy-efficient cellular networks via cell power and/or range adaptations.

The rest of this paper is organized as follows. Section \ref{sec:system model} introduces the system model, and derives the BS power scaling law under STPC. Section \ref{sec:optimal power and range adaptation} studies the joint cell range adaptation and LTPC problem. Section \ref{sec:suboptimal} presents various low-complexity suboptimal schemes. Section \ref{sec:performance comparison} compares the performance of optimal and suboptimal schemes through numerical examples. Finally, Section \ref{sec:conclusion} concludes the paper.

\section{System Model}\label{sec:system model}
We consider an OFDMA downlink in one particular cell with bandwidth $W$ Hz. It is assumed that the BS can adaptively adjust its cell coverage according to MU density and power budget through admission control. In this section, we first introduce a spatial model of cellular traffic based on MUs distributed according to a HPPP. Then, we elaborate on the proposed bandwidth sharing scheme for the OFDMA-based broadcast channel. Finally, we describe the STPC, based on which a power scaling law relating the (short-term) average transmit power at a BS given a pair of coverage range and MU density is derived.

\subsection{Traffic Model}
The two-dimensional Poisson process has been used to model the locations of MUs in a cellular network. In this paper, we assume that MUs form a HPPP $\Phi_m$ of density $\lambda_m$ in the Euclidean plane. Considering that every MU within the cell coverage requests connection (voice service or data application) randomly and independently with probability $q$, then according to the Marking Theorem \cite{Kingman93}, the active MUs (that need to communicate with a BS) form another HPPP $\Phi$ of density $\lambda$,\footnote{BS is assumed to support all MUs, within coverage, who request service.} where $\lambda = q\lambda_m$. Since we are interested in active MUs, we refer to active MUs simply as MUs in the rest of this paper. The MU density $\lambda$ is assumed to be a non-negative random variable with finite support, i.e. $0 \leq \lambda \leq \lambda_{\text{max}}$, with $f_{\lambda}(\cdot)$ and $F_{\lambda}(\cdot)$ denoting its probability density function (PDF) and cumulative distribution function (CDF), respectively. Let $N \triangleq |\Phi(B)|$ represent the total number of MUs within a cell, denoted by $B$. Then $N$ is a Poisson random variable with mean $\mu_N \triangleq \lambda\pi R^2$, where $R$ denotes the cell radius, and probability mass function (PMF)
\begin{align}\label{eq:poisson distribution}
\mbox{Pr}[N = n] = \frac{\mu_N^n}{n!}e^{-\mu_N}, ~~ n = 0,1, \ldots.
\end{align}

\subsection{Equal Bandwidth Sharing}

\begin{figure}
\centering
\epsfxsize=0.7\linewidth
    \includegraphics[width=9cm]{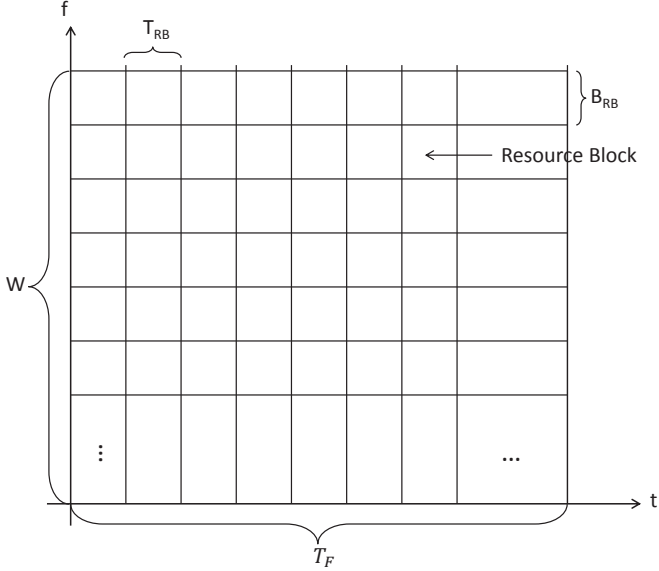}\\
  \caption{Equal bandwidth sharing (EBS)} \label{fig:DBS}
\end{figure}

Practically, dynamic bandwidth sharing (DBS) can be realized by users' time-sharing the available sub-carriers in OFDMA. To be more specific, the available time-frequency resource is divided into Resource Blocks (RBs) over both time and frequency, which are allocated among MUs such that each MU can be ideally assigned an effective bandwidth with arbitrary value from $0$ to $W$ Hz. Note that in general, DBS allocates the available RBs dynamically among MUs in order to optimize certain system-level utility (e.g. throughput) based on the number of MUs, their channels from the BS, and their QoS requirements. For the purpose of exposition, in this paper we assume a simplified equal bandwidth sharing (EBS) scheme among MUs, i.e., the effective bandwidth allocated to MU $i$, $i = 1,2 ..., N$, is $W/N$ Hz.

An illustration of the EBS within a scheduled transmission frame $T_F$ is shown in Fig. \ref{fig:DBS}. The available time-frequency resource is divided into RBs with dimensions $T_{\text{RB}}$ and $B_{\text{RB}}$ over time and frequency, respectively. $T_{\text{RB}}$ and $B_{\text{RB}}$ are assumed to be much smaller than the channel coherence time, $T_c$, and the channel coherence bandwidth, $B_c$, respectively; thus a flat-fading channel can be assumed in each RB. Let $N_F = \frac{W}{B_{\text{RB}}}$ and $N_T = \frac{T_F}{T_{\text{RB}}}$ be the number of frequency slices and time slices, respectively, within a transmission frame. The total number of available RBs within one frame can be computed as $U = N_FN_T$, which is assumed to be large enough such that each MU can be assigned a continuous effective bandwidth $\frac{U_i}{U}W$, where $U_i$ is the number of RBs allocated to MU $i$. For example, $4$ RB's are allocated to MU $i$ as shown in Fig. \ref{fig:DBS}. The total bandwidth allocated to MU $i$ is therefore $\frac{4}{N_F}W$, over a period of $N_T$ channel uses, where a channel use corresponds to $T_{\text{RB}}$ seconds. Therefore, MU $i$ is given $\frac{4W}{N_FN_T} = \frac{4W}{U}$ Hertz of bandwidth per channel use, which also implies that the BS is serving $N = \frac{U}{4}$ active MUs by EBS.

%For example, as shown by the shaded blocks in Fig. \ref{fig:DBS}, the effective bandwidth allocated to MU $i$ is $4W/U$, which also implies that the BS is serving $N = U/4$ active MUs by EBS.

With EBS, the achievable rate for MU $i$, given received signal power $S_i$, is
\begin{align}\label{eq:achievable rate}
V_i = \frac{W}{N}\log_2\left(1 + \frac{NS_i}{\Gamma N_0W}\right)
\end{align}
where $\Gamma$ accounts for the gap from the channel capacity due to a practical coding and modulation scheme, and $N_0$ is the power spectral density of the additive white Gaussian noise (AWGN).

Suppose that channel coding is performed over $L$ non-contiguous RBs allocated to a MU (c.f. Fig. \ref{fig:DBS} with $L = 4$). Then from (\ref{eq:achievable rate}), the average achievable rate of MU $i$ over $L \geq 1$ RBs is given by \cite{Tse05}
\begin{align}\label{eq:average achievable rate}
\bar{V}_i = \frac{1}{L}\sum_{l=1}^{L}\frac{W}{N}\log_2\left(1 + \frac{NS_{i, l}}{\Gamma N_0W}\right)
\end{align}
where $S_{i, l}$ is the received signal power at the $l$th allocated RB, $l = 1, ..., L$, and $S_{i, l}$'s are independent over $l$ due to independent channel fading if the $L$ RBs allocated to a MU are sufficiently far apart in time and/or frequency.

\subsection{Power Scaling Law}\label{sec:power scaling law}
We assume a simplified channel model consisting of distance-dependent pathloss with path loss exponent $\alpha > 2$ and an additional random term accounting for short-term fading of the channel from the BS to each MU. With the assumed channel model, the received signal power for the $l$th RB of MU $i$ is given by
\begin{align}\label{eq:path loss model}
S_{i,l} & = \left\{ \begin{array}{cl} \displaystyle
P_ih_{i,l}K\left(\frac{r_i}{r_0}\right)^{-\alpha} & \mbox{if } r_i \geq r_0 \\
P_ih_{i,l}K & \mbox{otherwise} \end{array}\right.
\end{align}
where $r_i$ is a random variable representing the distance between MU $i$ and BS, $K$ is a constant equal to the pathloss at a reference distance $r_0$, $h_{i,l}$ is an exponential random variable with unit mean accounting for Rayleigh fading with $h_{i, l}$'s being independent and identically distributed (i.i.d) over both $i$ and $l$, and $P_i$ is the transmit power for MU $i$, which is assumed to be identical for all $l$'s since the realizations of $h_{i, l}$'s are not assumed to be known at BS. It is easy to verify that $S_{i,l}$'s are i.i.d over $l$ as previously assumed.

To characterize the required minimum transmit power for MU $i$, $P_i$, outage performance is considered as the user-level QoS constraint. An outage event occurs when the link between MU $i$ and BS cannot support a desired target rate $\bar{v}$ bits/sec, which is assumed to be equal for all MUs for simplicity. According to (\ref{eq:average achievable rate}), the outage probability for MU $i$ is given by
\begin{align}\label{eq:outage probability}
\text{P}_{\text{out}}^i = \mbox{Pr}\left\{\sum_{l=1}^{L}\frac{W}{N}\log_2\left(1 + \frac{NS_{i,l}}{\Gamma N_0W}\right) < L\bar{v}\right\}.
\end{align}
Since outage typically occurs when none of the $L$ parallel channels can support the average rate $\bar{v}$ \cite{Tse05}, (\ref{eq:outage probability}) can be properly approximated as
\begin{align}\label{eq:approximate outage probability}
\text{P}_{\text{out}}^i & \approx \prod_{l = 1}^{L}\mbox{Pr}\left\{\frac{W}{N}\log_2\left(1 + \frac{NS_{i,l}}{\Gamma N_0W}\right) < \bar{v}\right\} \nonumber \\
& = \left(\mbox{Pr}\left\{ S_{i,1} < \frac{\Gamma N_0W}{N}(2^{\frac{N\bar{v}}{W}}-1)\right\}\right)^{L}.
\end{align}
Given $r_i$, $S_{i,1}$ is an exponential random variable with mean $\bar{S}_{i,1}$, which is given by
\begin{align}\label{eq:mean receive power}
\bar{S}_{i,1} & = \left\{ \begin{array}{cl} \displaystyle
P_iK\left(\frac{r_i}{r_0}\right)^{-\alpha} & \mbox{if } r_i \geq r_0 \\
P_iK & \mbox{otherwise}. \end{array}\right.
\end{align}
Thus, the outage probability for MU $i$ given distance from BS $r_i$ can be simplified as
\begin{align}\label{eq:simplified outage probability}
\text{P}_{\text{out}}^i(r_i) \approx \left[1 - \exp\left(-\frac{\Gamma N_0W}{N\bar{S}_{i,1}}(2^{\frac{N\bar{v}}{W}}-1)\right)\right]^{L}.
\end{align}

Let $\bar{\text{P}}_{\text{out}}$ denote the maximum allowable outage probability for all MUs. Then the inequality
\begin{align}\label{eq:quality of service demaind}
\text{P}_{\text{out}}^i \leq \bar{\text{P}}_{\text{out}}
\end{align}
needs to be maintained for all $i$'s. From (\ref{eq:mean receive power}), (\ref{eq:simplified outage probability}) and (\ref{eq:quality of service demaind}), we can obtain $P_i$ given $r_i$ and $N$ for the BS's STPC as\footnote{Note that several other quantities such as $\bar{V}_i$ and $\text{P}^i_{\text{out}}$ are also dependent on $N$, but to simplify notation, we did not explicitly display this dependency when defining them. However, the manipulations of $P_i$ to follow do involve $N$ and therefore we write $P_i$ as a function of $r_i$ and $N$ below.}
\begin{align}\label{eq:single user transmission power}
P_i(r_i, N) & = \left\{ \begin{array}{cl} \displaystyle
\frac{\Gamma N_0W}{KC_1}\cdot\frac{2^{NC_2} - 1}{N}\cdot\frac{r_i^{\alpha}}{r_0^{\alpha}} & \mbox{if } r_i \geq r_0 \\
\frac{\Gamma N_0W}{KC_1}\cdot\frac{2^{NC_2} - 1}{N} & \mbox{otherwise} \end{array}\right.
\end{align}
where $C_1 = -\ln(1 - \bar{\text{P}}_{\text{out}}^{1/L})$ and $C_2 = \frac{\bar{v}}{W}$. With $P_i(r_i, N)$, the total transmit power $P_t$ at the BS can be expressed as
\begin{align}\label{eq:general total power}
P_t = \sum_{i = 1}^{N} P_i(r_i, N).
\end{align}
Note that $P_t$ is a random variable due to the randomness in the number of MUs, $N$, and their random distances from the BS, $r_i$'s.

In this paper, we assume that the BS can perform a slow LTPC based on the MU density variation, in addition to the more rapid STPC, for the purpose of minimizing the long-term energy consumption (more details will be given in Section \ref{sec:optimal power and range adaptation}). Considering the fluctuations of $P_t$ given coverage range $R$ and MU density $\lambda$, according to (\ref{eq:general total power}), a power scaling law that averages the random effects of the number of MUs and their locations is desired to facilitate the LTPC design to be studied in Section \ref{sec:optimal power and range adaptation}. This motivates us to find the (short-term) average transmit power $\bar{P}_t \triangleq \mathbb{E}[P_t]$ at BS for a given pair of $R$ and $\lambda$, where the expectation is taken over $N$ and $r_i$'s.

The approach for finding $\bar{P}_t$ is to apply the law of iterated expectations, i.e.,
\begin{align}\label{eq:conditonal expectation formula}
\bar{P}_t = \mathbb{E}_{N}\left[\mathbb{E}[P_t|N]\right]
\end{align}
where the inner expectation is taken over the random user locations given $N = n$ number of MUs, and the outer expectation is performed over the Poisson distributed $N$. This method works because $\mathbb{E}[P_t|N=n]$ in (\ref{eq:conditonal expectation formula}) can be obtained using the following property of conditioned HPPP \cite{Kingman93}:
\begin{align}\label{eq:conditioned expected transmision power formular}
\mathbb{E}[P_t|N=n] & = \mathbb{E}\left[\sum_{i = 1}^{n} P_i(r_i, n)\right] \nonumber \\
& = n\mathbb{E}[P_i(r_i, n)]
\end{align}
where $P_i(r_i, n)$ represents the required transmit power from the BS to any MU $i$ with distance $r_i$ given that $N = n$ number of MUs equally share the total bandwidth $W$ by EBS. It can be further verified that given $N = n$, MU $i$ is uniformly distributed within a circular coverage area with radius $R$. Thus, $\mathbb{E}[P_i(r_i, n)]$ is identical for all $i$'s, and computed as
\begin{align}\label{eq:pdf of location}
\mathbb{E}[P_i(r_i, n)] = \int_0^R P_i(r_i, n)f(r_i) dr_i
\end{align}
where $f(r_i) = \frac{2r_i}{R^2}$, $0 \leq r_i \leq R$, is the PDF of $r_i$.

Using (\ref{eq:pdf of location}) and averaging $\mathbb{E}[P_t|N=n]$ in (\ref{eq:conditioned expected transmision power formular}) over the Poisson distribution of $N$, we obtain a closed-form expression for $\bar{P}_t$, which is given in the following theorem.
\begin{theorem}\label{theorem:0}
Consider an OFDMA-based broadcast channel, where the available bandwidth $W$ Hz is equally shared among all MUs with STPC to support a target rate $\bar{v}$ bits/sec with outage constraint $\bar{\text{P}}_{\text{out}}$. Suppose that the channels from the BS to all MUs experience independent Rayleigh fading, then the transmit power at the BS averaged over MU population $N$ and BS-MU distance $r_i$, given a coverage range $R$ and a MU intensity $\lambda$, is approximated by
\begin{align}\label{eq:power radium relationship}
\bar{P}_t(R, \lambda) = D_1R^{\alpha}\left(2^{D_2\pi\lambda R^2} - 1\right)
\end{align}
where $D_1 = \frac{2\Gamma N_0W}{K(-\ln(1 - \bar{\text{P}}_{\text{out}}^{1/L}))(\alpha + 2)r_0^{\alpha}}$ and $D_2 = \frac{\bar{v}}{W}$ is the per-user spectrum efficiency in bps/Hz.
\end{theorem}
\begin{proof}
See Appendix \ref{appendix:proof theorem 0}.
\end{proof}

\begin{remark}\label{remark:0}
Theorem \ref{theorem:0} relates the average BS transmit power $\bar{P}_t$ with cell range $R$ and MU density $\lambda$. Given $R$, $\bar{P}_t$ grows exponentially with increasing $\lambda$ due to the reduced bandwidth equally allocated among (on average) $\mu_N = \lambda\pi R^2$ MUs. On the other hand, given $\lambda$, besides the exponential increment in $\bar{P}_t$ with respect to $R^2$ due to the similar effect of per-user bandwidth reduction, there exists an extra polynomial term $R^{\alpha}$ in $\bar{P}_t$, due to the increased power consumption needed to compensate for more significant path loss with growing $R$. Since $\bar{P}_t$ is a strictly increasing function of both $R$ and $\lambda$, to maintain a constant $\bar{P}_t$, $R$ needs to be reduced when $\lambda$ increases and vice versa. Theorem \ref{theorem:0} therefore quantifies the relationship among BS transmit power, cell size and MU density, which enables the design of the (long-term) cell adaptation strategies introduced in the rest of this paper.
\end{remark}

\section{Optimal Power and Range Adaptation}\label{sec:optimal power and range adaptation}
Power and range adaptation is the combined task of cell range adaptation and BS LTPC (including on-off control), which are both assumed to be performed on the time scale of MU density variation. Since MU's density variation is much slower as compared with MU's channel variation (which is taken care of by STPC studied in Section \ref{sec:power scaling law}), LTPC is implemented over $\bar{P}_t$ given in (\ref{eq:power radium relationship}) for the purpose of minimizing the BS's long-term energy consumption.

In this section, we first present a practical energy consumption model for BS by considering both transmission and non-transmission related power consumptions. Based on the presented energy consumption model, we study a joint cell range adaptation and LTPC problem to minimize the long-term power consumption at BS under a system-level throughput constraint.

\subsection{Energy Consumption Model at BS}\label{sec:energy consumption model}
The energy consumption of a BS in general includes two parts: transmit power $\bar{P}_{t}$ and a constant power $P_{c}$ accounting for all non-transmission related power consumption of e.g. electronic hardware and air conditioning. When the BS does not need to support any user, it can switch to a ``sleep'' mode \cite{Blume10}, by turning off the power amplifier to reduce energy consumption. We note that the two cases of $R > 0$ and $R = 0$ correspond to ``on'' and ``off (sleep)'' modes of BS, respectively. A power consumption model for the BS is thus given by
\begin{align}\label{eq:BS power consumption}
\bar{P}_{\text{BS}}(R, \lambda) & = \left\{ \begin{array}{cl} \displaystyle
a\bar{P}_{t}(R, \lambda) + P_{c}, &  ~ R > 0 \\
P_{\text{sleep}}, &  ~ R = 0 \end{array}\right.
\end{align}
where $\bar{P}_{\text{BS}}(R,\lambda)$ represents the (short-term) average power consumption at BS given a pair of $R$ and $\lambda$, $P_{\text{sleep}}$ denotes the power consumed during the off mode, and $a \geq 1$ corresponds to the scaling of the actual power consumed with the radiated power due to amplifier and feeder losses. In practice, $P_{\text{sleep}}$ is generally much smaller than $P_{c}$ \cite{3GppTR} and thus in this paper, we assume $P_{\text{sleep}} = 0$ for simplicity. Since $a$ is only a scaling constant, we further assume $a = 1$ in our subsequent analysis unless stated otherwise.

\subsection{Optimal Cell Adaptation}\label{sec:optimal scheme}
According to (\ref{eq:power radium relationship}), $\bar{P}_t(R, \lambda)$ is determined by $R$ and $\lambda$. LTPC is thus equivalent to range adaptation over $\lambda$, i.e., by first finding the range adaptation function $R(\lambda)$ and then obtaining $\bar{P}_t(R, \lambda)$ as $\bar{P}_t(R(\lambda), \lambda)$, the LTPC policy $\bar{P}_{\text{BS}}(R(\lambda), \lambda)$ follows from (\ref{eq:BS power consumption}). The joint cell range adaptation and LTPC problem can thus be formulated as
\begin{align}
\mathrm{(P0)}:~\mathop{\mathtt{Min.}}\limits_{R(\lambda) \geq 0} &
~~ \mathbb{E}_{\lambda}\left[\bar{P}_{\text{BS}}(R(\lambda), \lambda)\right]  \\
\mathtt{s.t.} & ~~ \mathbb{E}_{\lambda}\left[U(R(\lambda), \lambda)\right] \geq U_{\text{avg}} \label{eq:capacity constraint}\\
& ~~ \bar{P}_{\text{BS}}(R(\lambda), \lambda) \leq P_{\text{max}}, ~~ \forall \lambda \label{eq:max power constraint}
\end{align}
where $U(R(\lambda),\lambda)= \pi\lambda R^2(\lambda)$ corresponds to the (short-term) average number of supported MUs, $U_{\text{avg}}$ represents the (long-term) system throughput\footnote{Since a constant rate requirement $\bar{v}$ is assumed for all MUs and the effective system throughput equals to $\bar{v}U_{\text{avg}}(1-\bar{\text{P}}_{\text{out}})$, where $\bar{\text{P}}_{\text{out}}$ is a given outage probability target, the average number of supported MUs $U_{\text{avg}}$ is an equivalent measure of the effective system throughput.} constraint, and $P_{\text{max}}$ is the (short-term) power constraint at BS. For convenience, in the rest of this paper, $\bar{P}_t(R(\lambda),\lambda)$ and $\bar{P}_{\text{BS}}(R(\lambda),\lambda)$ are referred to as (short-term average) transmit power and power consumption at BS for a given $\lambda$, respectively, while $\mathbb{E}_{\lambda}\left[\bar{P}_t(R(\lambda),\lambda)\right]$ and $\mathbb{E}_{\lambda}\left[\bar{P}_{\text{BS}}(R(\lambda),\lambda)\right]$ are called the (long-term) average transmit power and average power consumption at BS, respectively.

Note that if choosing $R(\lambda)$ such that $\bar{P}_{\text{BS}}(R(\lambda),\lambda) = P_{\text{max}}$ for all $\lambda > 0$ still leads to a violation of constraint (\ref{eq:capacity constraint}), then Problem (P0) is infeasible. For analytical tractability, we only consider the case where $U_{\text{avg}}$ yields a feasible (P0). (P0) is not convex due to the non-convexity of both the objective function (at $R = 0$) and the throughput constraint (\ref{eq:capacity constraint}) since $U(R(\lambda), \lambda)$ is a non-concave function over $R(\lambda)$.

We start with reformulating (P0) via a change of variable: $x = R^2$, and making the constraint (\ref{eq:max power constraint}) implicit, which yields an equivalent problem
\begin{align}
\mathrm{(P1)}:~\mathop{\mathtt{Min.}}\limits_{x(\lambda) \in \mathcal{X}_a} &
~~ \mathbb{E}_{\lambda}\left[\bar{P}_{\text{BS}}(x(\lambda), \lambda)\right] \\
\mathtt{s.t.} & ~~ \mathbb{E}_{\lambda}\left[U(x(\lambda), \lambda)\right] \geq U_{\text{avg}} \label{eq:capacity constraint p1}
\end{align}
where $\mathcal{X}_a \triangleq \left\{x(\lambda) : x(\lambda) \geq 0, \bar{P}_{\text{BS}}(x(\lambda), \lambda) \leq P_{\text{max}}, \forall \lambda \right\}$. In (P1), the constraint (\ref{eq:capacity constraint p1}) becomes convex since $U(x(\lambda), \lambda) = \pi\lambda x(\lambda)$ is affine over $x(\lambda)$. Furthermore, $\mathcal{X}_a$ is a convex set, and $\mathbb{E}_{\lambda}\left[\bar{P}_{\text{BS}}(x(\lambda), \lambda)\right]$ is the affine mapping of an infinite number of quasi-convex functions $\bar{P}_{\text{BS}}(x(\lambda), \lambda)$ and can be shown to be quasi-convex. Therefore, (P1) is a quasi-convex optimization problem and it can be verified that Lagrangian duality method can be applied to solve (P1) globally optimally \cite{Boyd}. The Lagrangian of Problem (P1) is
\begin{align}
\mathcal{L}(x(\lambda), \mu) = \mathbb{E}_{\lambda}\left[\bar{P}_{\text{BS}}(x(\lambda), \lambda)\right] - \mu\left(\mathbb{E}_{\lambda}\left[U(x(\lambda), \lambda)\right] - U_{\text{avg}}\right)
\end{align}
where $\mu \geq 0$ is the dual variable associated with the throughput constraint (\ref{eq:capacity constraint p1}). Then it can be shown that solving (P1) is equivalent to solving parallel subproblems all having the same structure and each for a different value of $\lambda$. For a particular $\lambda$, the associated subproblem is expressed as
\begin{align}\label{eq:dual problem of p2 after decouple}
\mathop{\mathtt{Min.}}\limits_{x(\lambda) \in \mathcal{X}_a} &
~~ L_{\lambda}(x(\lambda), \mu)
\end{align}
where $L_{\lambda}(x(\lambda), \mu) = \bar{P}_{\text{BS}}(x(\lambda), \lambda) - \mu U(x(\lambda), \lambda)$.

To tackle the non-continuity of $\bar{P}_{\text{BS}}(x(\lambda), \lambda)$ at $x(\lambda) = 0$ (due to $P_{c} > P_{\text{sleep}} \triangleq 0$) and the power constraint $\bar{P}_{\text{BS}}(x(\lambda), \lambda) \leq P_{\text{max}}$, we first consider the case where BS is always on, i.e., $x(\lambda) > 0$ (thus, $\bar{P}_{\text{BS}}(x(\lambda), \lambda)$ is always differentiable) and there is no power constraint, i.e., $P_{\text{max}} = +\infty$. The power constraint and the non-continuity at $x(\lambda) = 0$ will be incorporated into the solution later without loss of optimality.

Denote $x_1^{*}(\lambda)$ and $x_2^{*}(\lambda)$ as the roots of the following two equations:
\begin{align}
\frac{\partial L_{\lambda}(x(\lambda), \mu)}{\partial x(\lambda)} & = 0, ~~  x(\lambda) > 0 \label{eq:optimal condition for always on}\\
\bar{P}_{\text{BS}}(x(\lambda), \lambda) & = P_{\text{max}}, \label{eq:radius controled by pmax}
\end{align}
respectively, where (\ref{eq:optimal condition for always on}) is the optimality condition for $x(\lambda)$ in the case where BS is always on with infinite power budget and (\ref{eq:radius controled by pmax}) gives the maximum coverage range due to finite $P_{\text{max}}$ for any given $\lambda$. {Note that it is difficult to obtain closed-form solutions for $x_1^{*}(\lambda)$ and $x_2^{*}(\lambda)$ due to the complex form of $\bar{P}_{\text{BS}}(x(\lambda), \lambda)$ in (\ref{eq:BS power consumption}). However, since $\bar{P}_{\text{BS}}(x(\lambda), \lambda)$ is a strictly increasing function of $x(\lambda)$, and furthermore is convex in $x(\lambda)$ when $x(\lambda) > 0$, $x_1^{*}(\lambda)$ and $x_2^{*}(\lambda)$ can both be obtained numerically by a simple bisection search given $\mu$ and/or $\lambda$.}

Let $x^{*}(\lambda)$ denote the optimal solution of Problem (\ref{eq:dual problem of p2 after decouple}) with finite $P_c$ and $P_{\text{max}}$. Then $x^{*}(\lambda)$ has three possible values: $x_1^{*}(\lambda)$, $x_2^{*}(\lambda)$ and 0, where $x_2^{*}(\lambda)$ is taken when $x_1^{*}(\lambda)$ violates the power constraint of $P_{\text{max}}$, i.e., $\bar{P}_{\text{BS}}(x_1^{*}(\lambda), \lambda) > P_{\text{max}}$. In the case of $\bar{P}_{\text{BS}}(x_1^{*}(\lambda), \lambda) \leq P_{\text{max}}$, a comparison between $L_{\lambda}(x_1^{*}(\lambda), \mu)$ and $L_{\lambda}(0, \mu) = 0$ is needed to tackle the non-continuity due to $P_c > 0$. If $L_{\lambda}(x_1^{*}(\lambda), \mu) < 0$, $x_1^{*}(\lambda)$ indeed gives the optimal solution; otherwise, we have $x^{*}(\lambda) = 0$ since it minimizes $L_{\lambda}(x(\lambda), \mu)$ over $x(\lambda) \geq 0$. On the other hand, if $\bar{P}_{\text{BS}}(x_1^{*}(\lambda), \lambda) > P_{\text{max}}$, a similar comparison between $L_{\lambda}(x_2^{*}(\lambda), \mu)$ and $L_{\lambda}(0, \mu) = 0$ is needed to verify the optimality between $x_2^{*}(\lambda)$ and $0$. Thus, the signs of $L_{\lambda}(x_1^{*}(\lambda), \mu)$ and $L_{\lambda}(x_2^{*}(\lambda), \mu)$ as well as the value of $\bar{P}_{\text{BS}}(x_1^{*}(\lambda), \lambda)$ jointly determine $x^{*}(\lambda)$, as summarized below:
\begin{align}\label{eq:original solution}
x^{*}(\lambda) & = \left\{ \begin{array}{cl} \displaystyle
x_1^{*}(\lambda) & \mbox{if } \left. \begin{array}{cl} \displaystyle \bar{P}_{\text{BS}}(x_1^{*}(\lambda), \lambda) \leq P_{\text{max}}, \\ L_{\lambda}(x_1^{*}(\lambda), \mu) < 0 \end{array}\right. \\
x_2^{*}(\lambda) & \mbox{if } \left. \begin{array}{cl} \displaystyle \bar{P}_{\text{BS}}(x_1^{*}(\lambda), \lambda) > P_{\text{max}}, \\ L_{\lambda}(x_2^{*}(\lambda), \mu) < 0 \end{array}\right. \\
0 & \mbox{otherwise}. \end{array}\right.
\end{align}

To avoid checking the conditions in (\ref{eq:original solution}) for all $\lambda$'s and gain more insights to the optimal power and range adaptation scheme, we proceed to characterize some critical values of $\lambda$, based on which the BS can determine $x^{*}(\lambda)$ with only the knowledge of the current density $\lambda$, through the following lemmas.
\begin{lemma}\label{lemma:0}
There exists $\lambda_1$, where $L_{\lambda}(x_1^{*}(\lambda_1), \mu) = 0$, such that $L_{\lambda}(x_1^{*}(\lambda), \mu)$ is positive for all $\lambda < \lambda_1$ and negative for all $\lambda > \lambda_1$.
\end{lemma}
\begin{proof}
See Appendix \ref{appendix:proof lemma 0}.
\end{proof}
\begin{lemma}\label{lemma:1}
$x_1^{*}(\lambda)$ is a strictly decreasing function of $\lambda$; $\bar{P}_{\text{BS}}(x_1^{*}(\lambda), \lambda)$ and $U(x_1^{*}(\lambda), \lambda)$ are all strictly increasing functions of $\lambda$.
\end{lemma}
\begin{proof}
See Appendix \ref{appendix:proof lemma 1}.
\end{proof}
\begin{lemma}\label{lemma:2}
$x_2^{*}(\lambda)$ is a strictly decreasing function of $\lambda$; $U(x_2^{*}(\lambda), \lambda)$ is a strictly increasing function of $\lambda$.
\end{lemma}
\begin{proof}
The monotonicity of $x_2^{*}(\lambda)$ can be directly obtained from Remark \ref{remark:0}. The proof for $U(x_2^{*}(\lambda), \lambda)$ is similar to that of Lemma \ref{lemma:1} in Appendix \ref{appendix:proof lemma 1}, and is thus omitted for brevity.
\end{proof}

Since $\bar{P}_{\text{BS}}(x_1^{*}(\lambda), \lambda)$ is a strictly increasing function of $\lambda$, there exists $\lambda_2$ with $\bar{P}_{\text{BS}}(x_1^{*}(\lambda_2), \lambda_2) = P_{\text{max}}$, above which $\bar{P}_{\text{BS}}(x_1^{*}(\lambda), \lambda) > P_{\text{max}}$. Furthermore, since $U(x_2^{*}(\lambda), \lambda)$ strictly increases with $\lambda$, $L_{\lambda}(x_2^{*}(\lambda), \mu) = P_{\text{max}} - \mu U(x_2^{*}(\lambda), \lambda)$ is thus a strictly decreasing function of $\lambda$ and there exists $\lambda_3$ with $L_{\lambda}(x_2^{*}(\lambda_3), \mu) =0$, such that $L_{\lambda}(x_2^{*}(\lambda), \mu) < 0$ for all $\lambda > \lambda_3$. Therefore, the conditions in (\ref{eq:original solution}) can be simplified as the inequalities among $\lambda_1$, $\lambda_2$ and $\lambda_3$, which is presented in the following theorem.
\begin{theorem}\label{theorem:1}
The optimal solution of Problem (P1) is given by
\begin{itemize}
        \item If $\lambda_2 \geq \lambda_1$
        \begin{align}
        x^{*}(\lambda) & = \left\{ \begin{array}{cl} \displaystyle
        0 & \mbox{if } \lambda \leq \lambda_1  \\
        x_1^{*}(\lambda) & \mbox{if } \lambda_1 < \lambda \leq \lambda_2  \\
        x_2^{*}(\lambda) & \mbox{otherwise}. \end{array}\right.
        \end{align}
        \item If $\lambda_2 < \lambda_1$
        \begin{align}
        x^{*}(\lambda) & = \left\{ \begin{array}{cl} \displaystyle
        0 & \mbox{if } \lambda \leq \lambda_3  \\
        x_2^{*}(\lambda) & \mbox{otherwise}. \end{array}\right.
        \end{align}
\end{itemize}
\end{theorem}
\begin{proof}
See Appendix \ref{appendix:proof theorem 1}.
\end{proof}

Note that Problem (P1) needs to be solved by iteratively solving $x^{*}(\lambda)$ with a fixed $\mu$ based on Theorem \ref{theorem:1}, and updating $\mu$ via the bisection search until the throughput constraint (\ref{eq:capacity constraint p1}) is met with equality. The optimal solution of Problem (P0), $R^{*}(\lambda)$, can then be obtained as $R^{*}(\lambda) = \sqrt{x^{*}(\lambda)}$. From Theorem \ref{theorem:1}, Lemma \ref{lemma:1} and Lemma \ref{lemma:2}, we obtain the following corollary.
\begin{corollary}\label{corollary:0}
$R^{*}(\lambda)$ and $U(R^{*}(\lambda),\lambda)$ are strictly decreasing and increasing functions of $\lambda$, respectively, if $R^{*}(\lambda) > 0$; $\bar{P}_{\text{BS}}(R^{*}(\lambda),\lambda)$ is a non-decreasing function of $\lambda$ if $R^{*}(\lambda) > 0$.
\end{corollary}
\begin{proof}
The proof directly follows from Lemmas \ref{lemma:1} and \ref{lemma:2}, and thus is omitted for brevity.
\end{proof}

Next, we illustrate the optimal solution $R^{*}(\lambda)$ to Problem (P0) to gain more insights to the optimal cell adaptation scheme. It is observed that there exists a cut-off value of $\lambda$ for each of the two cases in Theorem \ref{theorem:1}, below which the BS is switched off. This on-off behavior implies that allowing BS be switched off under light load is essentially optimal for energy saving. Since $x_2^{*}(\lambda)$ is the root of (\ref{eq:radius controled by pmax}), which corresponds to the maximum coverage range with finite $P_{\text{max}}$ for any given $\lambda$, it is worth noticing that when $\lambda_2 < \lambda_1$, constant power transmission with $P_{\text{max}}$ is optimal. The reason is that when $P_{\text{max}}$ is relatively small for the given throughput constraint $U_{\text{avg}}$, BS has to transmit at its maximum power at all the ``on'' time. According to Corollary \ref{corollary:0}, the average number of supported MUs $U(x^{*}(\lambda),\lambda)$ strictly increases with $\lambda$. This is because that under the optimal scheme, BS should support more MUs when the density is larger to optimize energy-efficiency.

\subsection{High Spectrum-Efficiency Regime}\label{sec:approximated optimal scheme}
Although Theorem \ref{theorem:1} reveals the structure of the optimal cell adaptation solution, which can be efficiently obtained numerically, the solution is expressed in terms of critical values of $\lambda$, namely $\lambda_1$, $\lambda_2$ and $\lambda_3$, for which closed-form expressions are difficult to be obtained. In this subsection, we obtain closed-form expressions of the solution in Theorem \ref{theorem:1} under a high spectrum-efficiency (HSE) assumption. It is observed from (\ref{eq:power radium relationship}) that $D_2\pi\lambda R^2 = \frac{\bar{v}\pi\lambda R^2}{W} = \frac{\bar{v}\mu_N}{W}$, which can be interpreted as the average network throughput in bps divided by the total bandwidth, and is thus the system spectrum-efficiency in bps/Hz. Therefore, the HSE assumption is equivalent to letting $D_2\pi\lambda R^2 \gg 1$. Under this condition, (\ref{eq:power radium relationship}) in Theorem \ref{theorem:0} can be simplified as
\begin{align}
\bar{P}_t(R, \lambda) = D_1R^{\alpha}2^{D_2\pi\lambda R^2}.
\end{align}

\begin{lemma}\label{lemma:3}
Under the HSE assumption of $D_2\pi\lambda R^2 \gg 1$, $x_1^{*}(\lambda)$ and $x_2^{*}(\lambda)$ in Theorem \ref{theorem:1} are given by
\begin{align}
x_1^{*}(\lambda) & = \frac{\alpha}{2D_3\pi\lambda}\mathcal{W}\left(\frac{2D_3\pi\lambda}{\alpha}\left(\frac{\mu}{D_1D_3}\right)^{\frac{2}{\alpha}}\right) \label{eq:x1 approximation}\\
x_2^{*}(\lambda) & = \frac{\alpha}{2D_3\pi\lambda}\mathcal{W}\left(\frac{2D_3\pi\lambda}{\alpha}\left(\frac{P^t_{\text{max}}}{D_1}\right)^{\frac{2}{\alpha}}\right) \label{eq:x2 approximation}
\end{align}
where $D_3 = (\ln2)D_2$, $P^t_{\text{max}} = P_{\text{max}}- P_c$, and $\mathcal{W}(\cdot)$ is the Lambert W function defined as $y = \mathcal{W}(y)e^{\mathcal{W}(y)}$ \cite{Wfunction}.
\end{lemma}
\begin{proof}
See Appendix \ref{appendix:proof lemma 3}.
\end{proof}

%\begin{figure}
%\centering
%\epsfxsize=0.7\linewidth
%    \includegraphics[width=12cm]{approximationverify.eps}\\
%  \caption{Optimal range in Theorem \ref{theorem:1} v.s. Approximated range in Lemma \ref{lemma:3}} \label{fig:approximation verification}
%\end{figure}

The accuracy of the above HSE approximation will be verified by numerical results in Section \ref{sec:performance comparison}. With (\ref{eq:x1 approximation}) and (\ref{eq:x2 approximation}), closed-form expressions of $U(x_1^{*}(\lambda),\lambda)$, $U(x_2^{*}(\lambda),\lambda)$ and $\bar{P}_{\text{BS}}(x_1^{*}(\lambda), \lambda)$ under the HSE assumption can be easily obtained, which can be verified to preserve the properties given in Lemmas \ref{lemma:0}-\ref{lemma:2} by using properties of the Lambert W function. For brevity, we omit the details here.

Moreover, we obtain the following corollary from Lemma \ref{lemma:3}.
\begin{corollary}\label{corollary:1}
Under the HSE assumption of $D_2\pi\lambda R^2 \gg 1$, $\lambda_1$, $\lambda_2$ and $\lambda_3$ in Theorem \ref{theorem:1} are given by
\begin{align}
\lambda_1 & = \left(\frac{1}{\pi D_3} + \frac{P_c}{\mu\pi}\right)\left(\frac{D_1D_3}{\mu}\right)^{\frac{2}{\alpha}}\exp\left(\frac{2}{\alpha} + \frac{2D_3P_c}{\mu\alpha}\right) \label{eq:approximate lambda1} \\
\lambda_2 & = \frac{\alpha P^t_{\text{max}}}{2\pi(\mu-D_3P^t_{\text{max}})}\left(\frac{D_1D_3}{\mu}\right)^{\frac{2}{\alpha}}\exp\left(\frac{D_3P^t_{\text{max}}}{\mu-D_3P^t_{\text{max}}}\right) \label{eq:approximate lambda2}\\
\lambda_3 & = \frac{P_{\text{max}}}{\mu\pi}\left(\frac{D_1}{P^t_{\text{max}}}\right)^{\frac{2}{\alpha}}\exp\left(\frac{2D_3P_{\text{max}}}{\mu\alpha}\right). \label{eq:approximate lambda3}
\end{align}
\end{corollary}
\begin{proof}
The proof is similar to that of Lemma \ref{lemma:3}, and thus omitted for brevity.
\end{proof}

\begin{remark}\label{remark:1}
$\lambda_1$, $\lambda_2$ and $\lambda_3$ in Corollary \ref{corollary:1} can be verified to be all strictly decreasing functions of the dual variable $\mu$ as follows. Let $\mu^{*}$ be the optimal dual solution of Problem (P1), $\lambda^{*}_1$, $\lambda^{*}_2$ and $\lambda^{*}_3$ be the corresponding critical values of $\lambda$ when $\mu = \mu^{*}$. Since $\mu^{*}$ strictly increases as the throughput constraint $U_{\text{avg}}$ increases, it follows from (\ref{eq:approximate lambda1})-(\ref{eq:approximate lambda3}) that $\lambda^{*}_1$, $\lambda^{*}_2$ and $\lambda^{*}_3$ are all strictly decreasing functions of $U_{\text{avg}}$. Since in Theorem \ref{theorem:1}, $\lambda_1$ and $\lambda_3$ are the thresholds of the MU density above which BS switches from off to on mode, their decrease with increasing $U_{\text{avg}}$ implies that BS needs to be stay on for more time if large system throughput is required.
\end{remark}

\section{Suboptimal schemes}\label{sec:suboptimal}
The optimal power and range adaptation policy presented in Section \ref{sec:optimal power and range adaptation} combines cell range adaptation and BS LTPC (including on-off control), suggesting that the energy saving at BS essentially comes from two major energy saving mechanisms (ESMs): range adaptation and BS on-off control. In this section, we propose four low-complexity suboptimal schemes, which can be considered as suboptimal solutions of (P0) with various combinations of these two ESMs, to investigate their effects on the system energy consumption.
\begin{enumerate}
\item {\bf Fixed range with BS on-off control (FRw/OFC)}: In this scheme, BS is switched off when MU density is lower than a cutoff value $\lambda_{c}$, while the coverage range $R$ is fixed as $R_{f}$ whenever BS is on. For a given $\lambda_c$, since from (\ref{eq:power radium relationship}) the BS transmission power is a strictly increasing function of $R$, $R_f$ should be chosen as the minimum value, denoted by $R_{f}(\lambda_c)$, to satisfy the throughput constraint $U_{\text{avg}}$ by applying BS power control with fixed coverage based on $\lambda$ according to (\ref{eq:power radium relationship}). Furthermore, $\lambda_c$ should be optimized to minimize the average BS power (including both transmission and non-transmission related portions) consumption. The optimal cutoff value $\lambda^{*}_c$ and its corresponding coverage range $R_f(\lambda^{*}_c)$ can be found via solving Problem (P0) by assuming the following (suboptimal) range adaptation policy:
    \begin{align}\label{eq:FRw/oFC constraint}
        R(\lambda) & = \left\{ \begin{array}{cl} \displaystyle
        R_f(\lambda_c) & \mbox{if } \lambda \geq \lambda_c  \\
        0 & \mbox{otherwise}. \end{array}\right.
    \end{align}
    Specifically, we have
    \begin{align}\label{eq:find optimal lambda_c for FRwOFC}
    \lambda^{*}_c =\arg\mathop{\mathtt{min.}}\limits_{\lambda_c < \lambda_{\text{max}}} \mathbb{E}_{\lambda_c}\left[\bar{P}_{\text{BS}}(R_f(\lambda_c), \lambda)\right]
    \end{align}
    where
    \begin{align}\label{eq:find fixed range}
    R_f(\lambda_c) = \mathop{\mathtt{min.}} & ~ R_f \\
    \mathtt{s.t.} & ~ \mathbb{E}_{\lambda_c}\left[U(R_f, \lambda)\right] \geq U_{\text{avg}} \nonumber \\
    & ~ \bar{P}_{\text{BS}}(R_f, \lambda) \leq P_{\text{max}}, \forall \lambda \geq \lambda_c. \nonumber
    \end{align}
    where $\mathbb{E}_{\lambda_c}\left[f(\lambda)\right] \triangleq \mathbb{E}_{\lambda}\left[\left.f(\lambda) \right|\lambda \geq \lambda_c\right]\mbox{Pr}\left\{\lambda \geq \lambda_c\right\}$. For a given $\lambda_c$, since $\mathbb{E}_{\lambda_c}\left[U(R_f, \lambda)\right]$ is a strictly increasing function of $R_f$, Problem (\ref{eq:find fixed range}) can be solved efficiently through the bisection search. Then, the optimal cut-off threshold in (\ref{eq:find optimal lambda_c for FRwOFC}) can be found by a line search over $[0, \lambda_{\text{max}}]$.
\item {\bf Fixed range without BS on-off control (FRw/oOFC)}: In this scheme, BS is not allowed to be switched off during operation. The coverage range is fixed as $R_{f}$, which is chosen as the minimum value of $R$ to satisfy the throughput constraint $U_{\text{avg}}$ by applying BS power control only based on $\lambda$ according to (\ref{eq:power radium relationship}). Note that FRw/oOFC can be treated as a special case of FRw/OFC with $\lambda_c$ in (\ref{eq:FRw/oFC constraint}) set to be $0$. Thus, the fixed coverage $R_f$ can be directly determined by solving Problem (\ref{eq:find fixed range}) with $\lambda_c = 0$.
\item {\bf Adaptive range with BS on-off control (ARw/OFC)}: In this scheme, BS is switched off when MU density is lower than a cutoff value $\lambda_c$, while BS transmits with constant power $P_{f} - P_c$ whenever it is powered on by applying range adaptation only based on $\lambda$ according to (\ref{eq:power radium relationship}). Given $P_{f}$, the corresponding $\lambda_c$ is chosen as the maximum value of $\lambda$, denoted by $\lambda_c(P_f)$, to satisfy the throughput constraint $U_{\text{avg}}$, in order to minimize the BS average power consumption $\mathbb{E}_{\lambda_c(P_f)}\left[P_{f}\right]$; $P_{f}$ is then optimized to further minimize the average power consumption at BS. The optimal transmit power $P^{*}_{f} - P_c$ and its corresponding cutoff value $\lambda_c(P^{*}_{f})$ can be obtained via solving Problem (P0) by assuming the following (suboptimal) range adaptation policy:
    \begin{align}\label{eq:ARw/oFC constraint}
        R(\lambda) & = \left\{ \begin{array}{cl} \displaystyle
        \bar{P}^{-1}_{\text{BS}}(P_f, \lambda) & \mbox{if } \lambda \geq \lambda_c(P_f)  \\
        0 & \mbox{otherwise}, \end{array}\right.
    \end{align}
    where $\bar{P}^{-1}_{\text{BS}}(P_f, \lambda)$ is the inverse function of (\ref{eq:BS power consumption}) which computes the coverage range with given BS power consumption $P_f$ and MU density $\lambda$. Specifically, we have
    \begin{align}\label{eq:find optimal power for ARw/OFC}
    P^{*}_{f} = \arg\mathop{\mathtt{min.}}\limits_{P_{f} \leq P_{\text{max}}}  \mathbb{E}_{\lambda_c(P_f)}\left[P_{f}\right]
    \end{align}
    where
    \begin{align}\label{eq:ARw/FTP}
    \lambda_c(P_{f}) = \mathop{\mathtt{max.}} & ~ \lambda_c \\
    \mathtt{s.t.} & ~ \mathbb{E}_{\lambda_c}\left[U(R(\lambda), \lambda)\right] \geq U_{\text{avg}} \nonumber \\
    & ~ \bar{P}_{\text{BS}}(R(\lambda), \lambda) = P_{f}, \forall \lambda \geq \lambda_c. \nonumber
    \end{align}
    Note that from (\ref{eq:ARw/oFC constraint}) and Remark \ref{remark:0}, $R(\lambda)$ increases strictly with $P_f$ given $\lambda$, $U(R(\lambda), \lambda) = \pi\lambda R^2(\lambda)$ is thus a strictly increasing function of $P_f$. Therefore, Problem (\ref{eq:ARw/FTP}) can be solved efficiently through the bisection search. Then, the optimal constant BS power consumption in (\ref{eq:find optimal power for ARw/OFC}) can be found by a line search over $[0, P_{\text{max}}]$.
\item {\bf Adaptive range without BS on-off control (ARw/oOFC)}: In this scheme, BS transmits with constant power $P_{f} - P_c$ and is not allowed to be switched off during operation, i.e., no BS power control is applied. The constant transmit power $P_{f} - P_c$ is chosen as the minimum value to satisfy the throughput constraint $U_{\text{avg}}$ by applying range adaptation only based on $\lambda$ according to (\ref{eq:power radium relationship}). Note that ARw/oOFC is a special case of ARw/OFC with $\lambda_c$ in (\ref{eq:ARw/oFC constraint}) set to be $0$. Thus, $P_f$ can be obtained by solving Problem (\ref{eq:find optimal power for ARw/OFC}) with $\lambda_c = 0$.
\end{enumerate}

The suboptimal schemes presented above all yield feasible and in general suboptimal solutions of Problem (P0). In particular, FRw/OFC and ARw/oOFC apply only BS power control (including  on-off control) and only range adaptation, respectively; ARw/OFC applies both BS on-off control and range adaptation, while FRw/oOFC does not apply any of them for lowest complexity. By comparing the performance of these suboptimal schemes with the optimal scheme presented in Section \ref{sec:optimal power and range adaptation}, we can investigate the effect of each individual ESM, namely, BS
power control and range adaptation on the BS energy saving, as will be shown in the next section through numerical examples.

\section{Numerical Results}\label{sec:performance comparison}
To obtain numerical results, we assume a time-varying traffic density with PDF: $f(\lambda) = \frac{4\lambda}{\lambda^2_{\text{max}}}, ~ 0 \leq \lambda \leq \frac{\lambda_{\text{max}}}{2}$; $f(\lambda) =\frac{4}{\lambda_{\text{max}}} - \frac{4\lambda}{\lambda^2_{\text{max}}}, ~ \frac{\lambda_{\text{max}}}{2} < \lambda \leq \lambda_{\text{max}}$, where $\lambda_{\text{max}} = 1\times10^{-4}$ MUs/$\mbox{m}^2$ is the peak traffic load. We consider pathloss and Rayleigh fading for channels between BS and MUs, where the pathloss exponent $\alpha$ is 3 and the outage probability threshold $\bar{\text{P}}_{\text{out}}$ is $10^{-3}$. The bandwidth $W$ and the rate requirement $\bar{v}$ of each MU are set to be 5 MHz and $150$ kbits/sec, respectively, if not specified otherwise. We also set a short-term power constraint at BS as $P_{\text{max}} = 160$ W. Other parameters are set as $\Gamma = 1$, $N_0=-174$ dBm/Hz, $r_0 = 10$ m, and $K= -60$ dB.

\begin{figure}
\centering
\epsfxsize=0.7\linewidth
    \includegraphics[width=8.8cm]{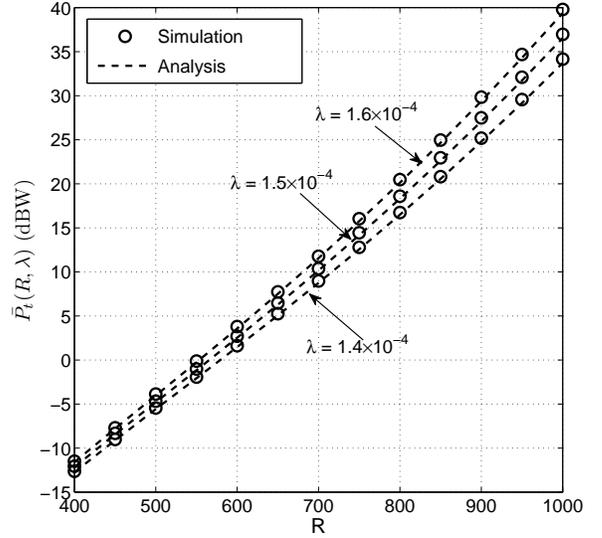}
\caption{Average transmit power $\bar{P_t}(R,\lambda)$ in Theorem \ref{theorem:0}.} \label{fig:power scaling law}
\end{figure}

Fig. \ref{fig:power scaling law} verifies the power scaling law in Theorem \ref{theorem:0}. For a given MU density $\lambda$, it is observed that the simulation results match well with our analytical result in (\ref{eq:power radium relationship}).

\begin{figure}
\centering
\subfigure[]{
\epsfxsize=0.7\linewidth
\includegraphics[width=8.8cm]{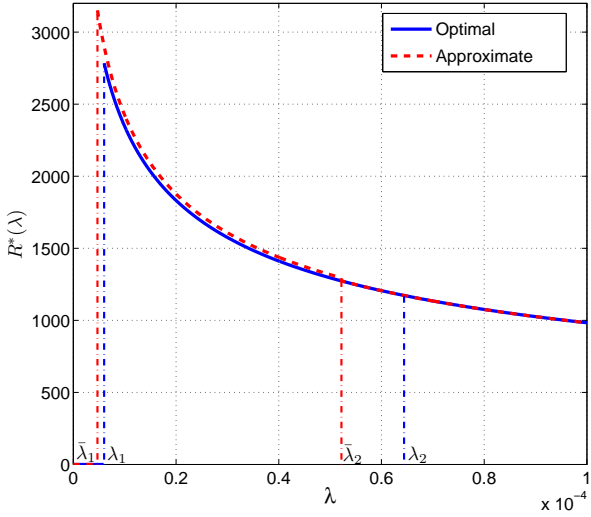}
\label{fig:optimaldemofirstRadius}
}
\subfigure[]{
\epsfxsize=0.7\linewidth
\includegraphics[width=8.8cm]{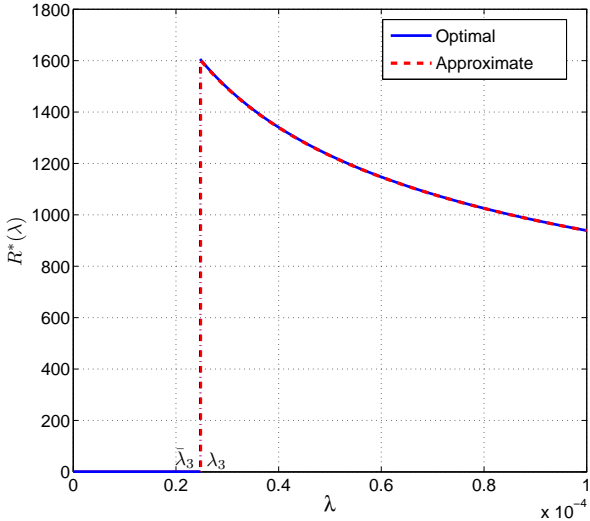}
\label{fig:optimaldemosecondRadius}
}
\caption[]{Optimal and approximate cell range adaptation v.s. MU density: (a) $\lambda_2 \geq \lambda_1$; (b) $\lambda_2 < \lambda_1$.}
\label{fig:Radius}
\end{figure}

\begin{figure}
\centering
\subfigure[]{
\epsfxsize=0.7\linewidth
\includegraphics[width=8.8cm]{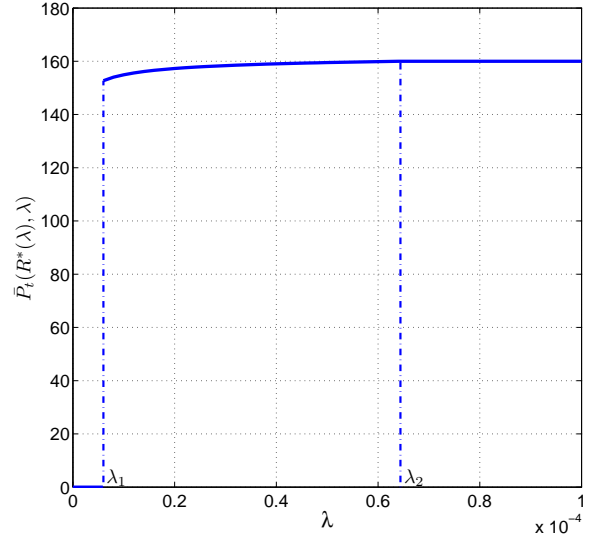}
\label{fig:optimaldemofirstPower}
}
\subfigure[]{
\epsfxsize=0.7\linewidth
\includegraphics[width=8.8cm]{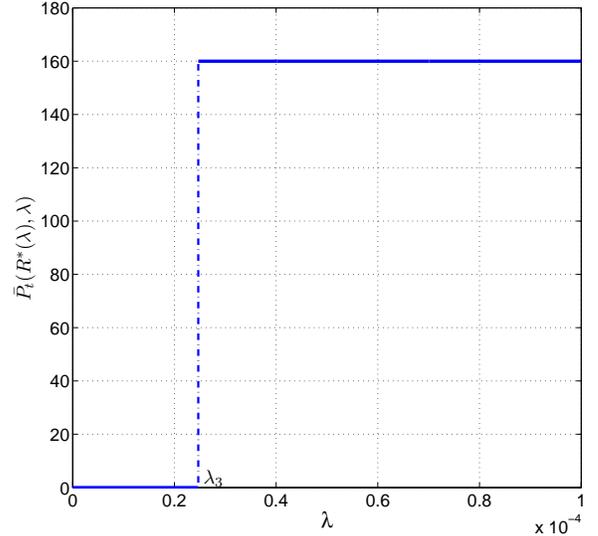}
\label{fig:optimaldemosecondPower}
}
\caption[]{Optimal BS power control v.s. MU density: (a) $\lambda_2 \geq \lambda_1$; (b) $\lambda_2 < \lambda_1$.}
\label{fig:Power}
\end{figure}

\begin{figure}
\centering
\subfigure[]{
\epsfxsize=0.7\linewidth
\includegraphics[width=8.8cm]{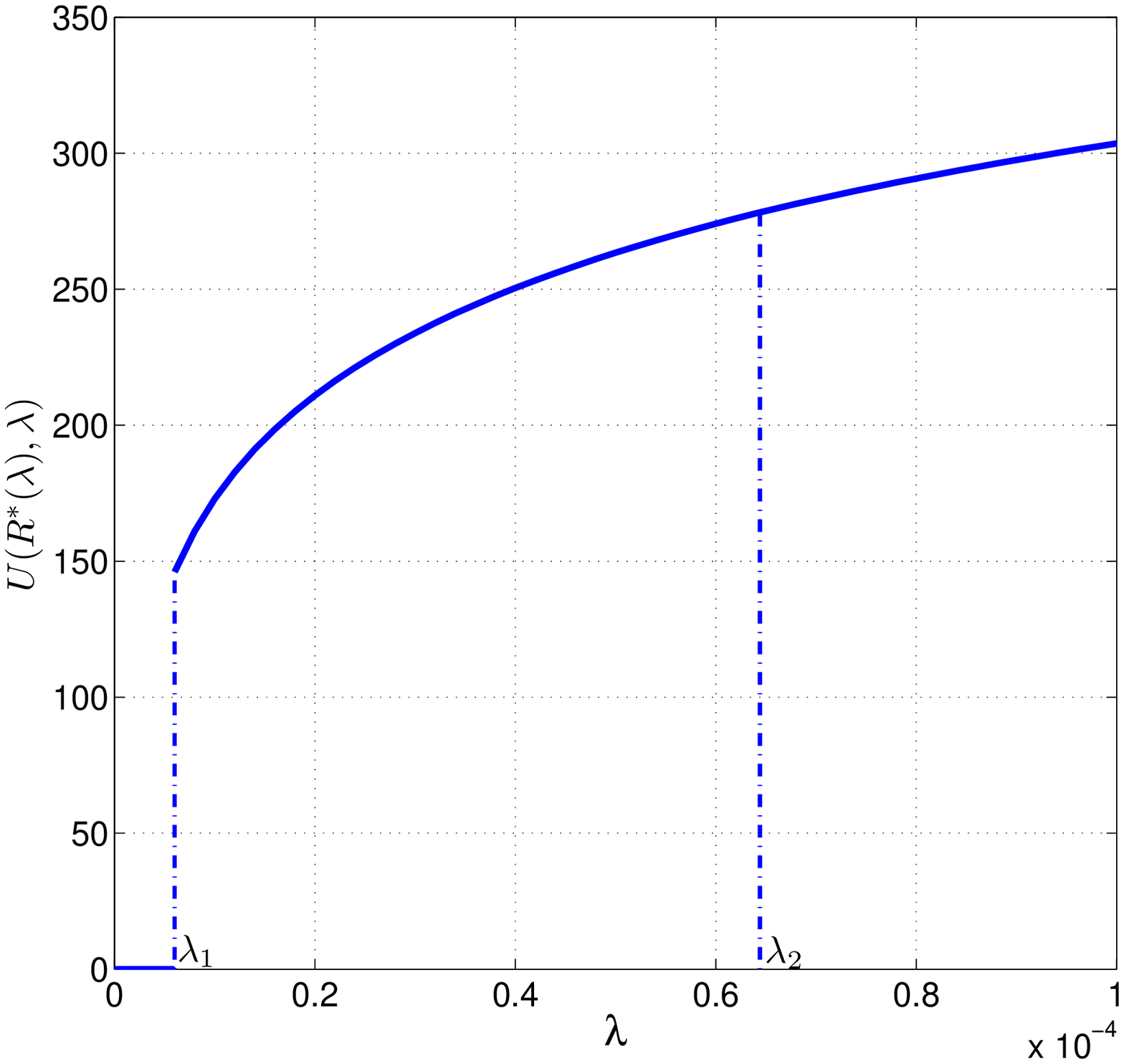}
\label{fig:optimaldemofirstUser}
}
\subfigure[]{
\epsfxsize=0.7\linewidth
\includegraphics[width=8.8cm]{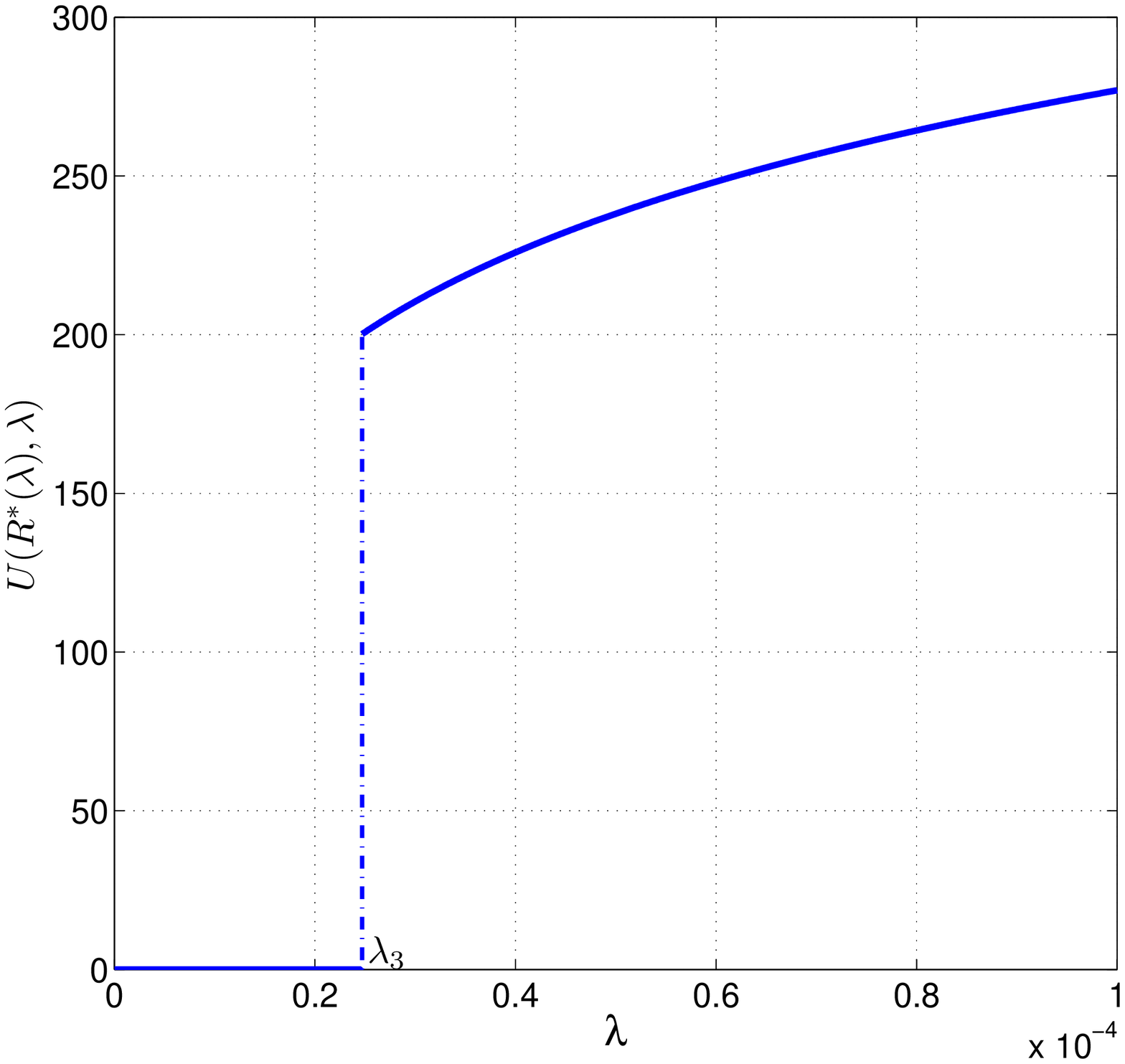}
\label{fig:optimaldemosecondUser}
}
\caption[]{Average number of supported users v.s. MU density: (a) $\lambda_2 \geq \lambda_1$; (b) $\lambda_2 < \lambda_1$.}
\label{fig:User}
\end{figure}

Fig. \ref{fig:optimaldemofirstRadius} and Fig. \ref{fig:optimaldemosecondRadius} show the optimal range adaptation in Theorem \ref{theorem:1} and the approximate range adaptation in Lemma \ref{lemma:3} under the HSE assumption as functions of MU density, i.e., $R^{*}(\lambda) = \sqrt{x^{*}(\lambda)}$, for the two cases of $\lambda_2 \geq \lambda_1$ and $\lambda_2 < \lambda_1$, respectively. Fig. \ref{fig:Power} and Fig. \ref{fig:User} show the corresponding optimal BS power adaptation and the resulting system throughput (in terms of average number of supported MUs), respectively\footnote{Since the results by the approximate range adaptation are almost no different from those in Figs. \ref{fig:Power} and \ref{fig:User}, we do not show them in these two figures for brevity.}. For Fig. \ref{fig:optimaldemofirstRadius}, Fig. \ref{fig:optimaldemofirstPower} and Fig. \ref{fig:optimaldemofirstUser}, it is assumed that $P_c = 120$ W and the corresponding optimal dual solution for Problem (P1) is $\mu^{*} = 1.05$, with which it can be verified that $\lambda_2 > \lambda_1$, i.e., corresponding to the first case in Theorem \ref{theorem:1}. For Fig. \ref{fig:optimaldemosecondRadius}, Fig. \ref{fig:optimaldemosecondPower} and Fig. \ref{fig:optimaldemosecondUser}, it is assumed that $P_c = 140$ W and $\mu^{*} = 0.8$; thus the critical values of $\lambda$ satisfy $\lambda_3 > \lambda_1 > \lambda_2$, which is in accordance with the second case of Theorem \ref{theorem:1}. It is observed that the numerical examples validate our theoretical results. As shown in Fig. \ref{fig:Radius}, a cut-off value of $\lambda$ exists (note that $\bar{\lambda}_i, i=1,2,3$, represent the approximate critical values of $\lambda$ obtained by Corollary \ref{corollary:1}) in either of the two cases of Theorem \ref{theorem:1}, which implies that allowing BS to be switched off under light load is optimal for energy saving. Note that from Fig. \ref{fig:Radius}, the approximate range adaptation is observed to match well with the optimal range adaptation for both cases. Fig. \ref{fig:Power} shows the optimal BS power adaptation versus the MU density. It is observed that once the BS is on, it transmits near or at the maximum power budget, which implies that constant power transmission at ``on'' mode is near or even optimal. This also explains the observation in Fig. \ref{fig:optimaldemofirstRadius} that the deviation of the approximated value of $\lambda_2$ or $\bar{\lambda}_2$ from $\lambda_2$ does not affect the accuracy of the approximate range adaptation policy, since the accuracy of $\lambda_1$ and $\lambda_3$ that control BS's on-off behavior is more crucial. The variations of the system throughput $U(R^{*}(\lambda),\lambda)$ with MU density $\lambda$ under the optimal scheme is shown in Fig. \ref{fig:User}. As discussed in Corollary \ref{corollary:0}, $U(R^{*}(\lambda),\lambda)$ is observed to increase strictly with $\lambda$ indicating that the optimal adaptation scheme takes advantage of higher MU density to maximize the system throughput.

\begin{figure}
\centering
\epsfxsize=0.7\linewidth
    \includegraphics[width=8.8cm]{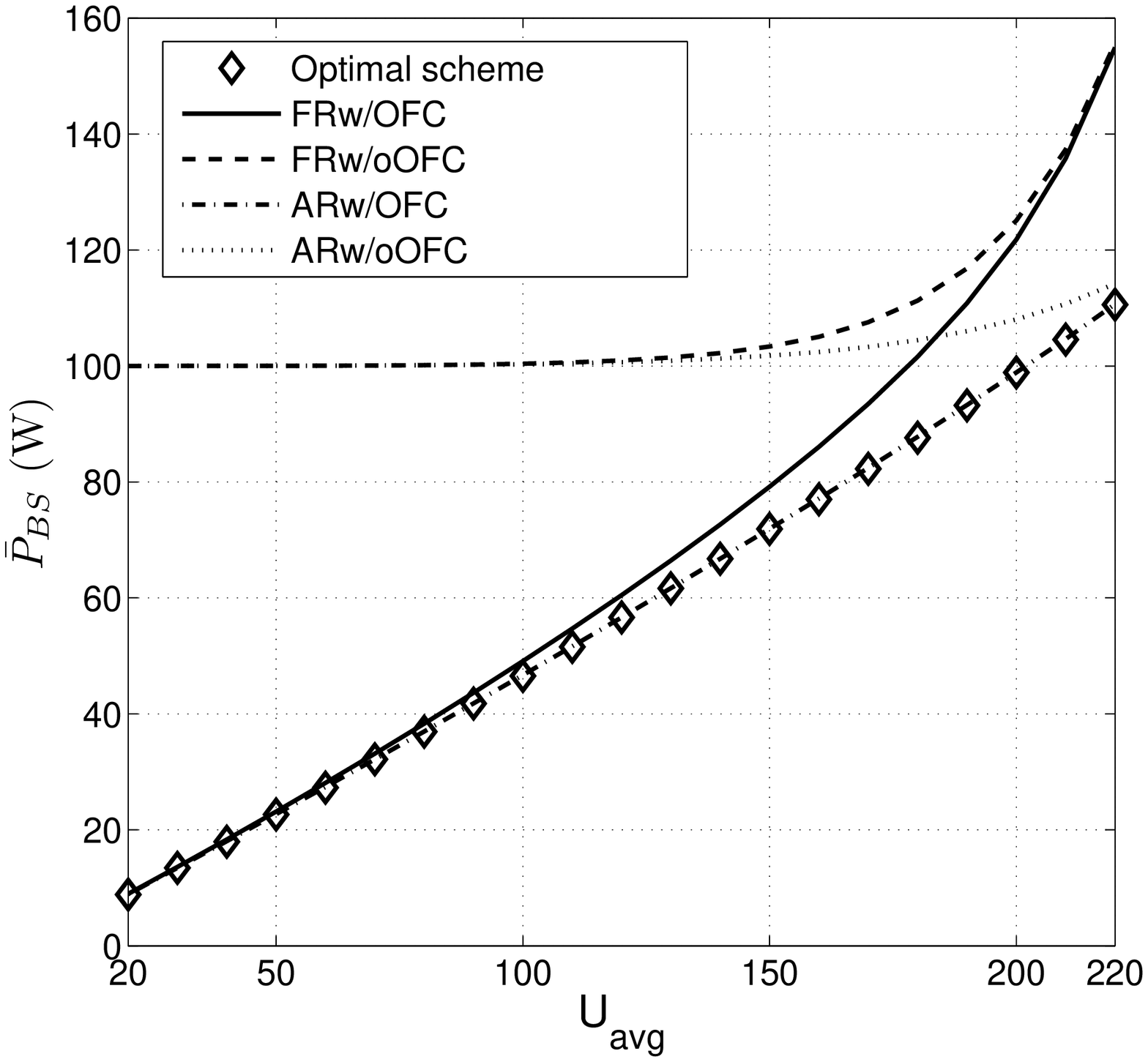}
\caption{Performance comparison with $P_c = 60$ W and $\bar{v} = 150$ Kbps} \label{fig:5_100_150_symmetric}
\end{figure}

Next, we compare the suboptimal schemes in Section \ref{sec:suboptimal} with the optimal scheme. With $P_c =60$ W, Fig. \ref{fig:5_100_150_symmetric} shows the average power consumption $\bar{P}_{\text{BS}}$ at BS versus the system throughput $U_{\text{avg}}$. From Fig. \ref{fig:5_100_150_symmetric}, we observe that ARw/OFC performs almost the same as the optimal scheme over the entire range of values of $U_{\text{avg}}$. This is because that constant power transmission at BS ``on'' mode is near or even optimal (c.f. Fig. \ref{fig:optimaldemosecondPower}) and ARw/OFC differs from the optimal scheme only in that the (long-term) transmit power control when BS is on (c.f. Fig. \ref{fig:optimaldemofirstPower} with $\lambda_1 < \lambda < \lambda_2$) is not implemented. It is also observed that when $U_{\text{avg}}$ is small, FRw/OFC has similar energy consumption as the optimal scheme and ARw/OFC; however, their performance gap is enlarged as $U_{\text{avg}}$ increases. A similar observation can be made by comparing ARw/oOFC and FRw/oOFC. From these observations, it follows that BS on-off control is the most effective ESM when the network throughput is low, while range adaptation plays a more important role when the network throughput becomes higher. Finally, we observe that ARw/OFC and FRw/OFC converge to ARw/oOFC and FRw/oOFC, respectively, as $U_{\text{avg}}$ increases. This is because that to achieve higher network throughput, BS needs to be ``on'' for more time to support larger number of MUs; as a result, BS on-off control is less useful for energy saving.

\begin{figure}
\centering
\epsfxsize=0.7\linewidth
    \includegraphics[width=8.8cm]{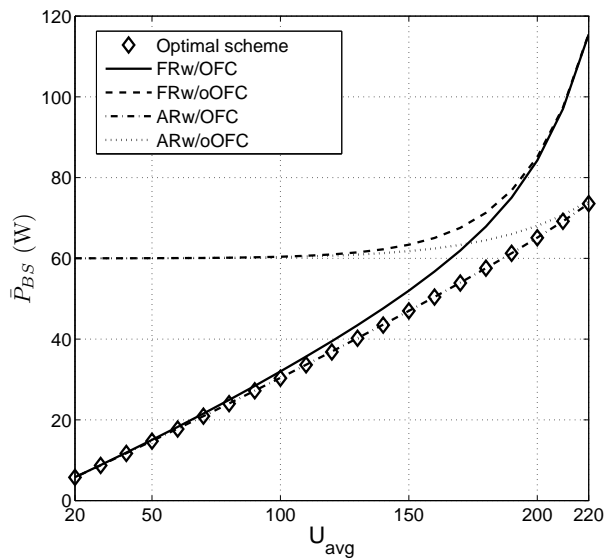}
\caption{Performance comparison with $P_c = 100$ W and $\bar{v} = 150$ Kbps} \label{fig:5_60_150_symmetric}
\end{figure}

In Fig. \ref{fig:5_60_150_symmetric}, we set $P_c = 100$ W to further evaluate the performances of different schemes under a higher non-transmission related power consumption at BS. Similar observations can be made from Fig. \ref{fig:5_60_150_symmetric} as in Fig. \ref{fig:5_100_150_symmetric}. However, it is worth noticing that BS on-off control plays a more dominant role for energy saving when $U_{\text{avg}}$ is small, since a higher $P_c$ is required. It is also interesting to observe that the performance gaps among different schemes with and without range adaptation are almost invariant to the change of $P_c$ at high network throughput, which is around $45$ W in both Figs. \ref{fig:5_100_150_symmetric} and \ref{fig:5_60_150_symmetric} with $U_{\text{avg}} = 220$. In Fig. \ref{fig:5_100_500_symmetric}, $P_c$ is reset as $60$ W but the transmission rate for each MU $\bar{v}$ is increased to $500$ kbits/sec to model the case with high-rate multimedia traffic. The simulation result shows that the convergence between different schemes with and without BS on-off control is much faster, which implies that range adaptation becomes more effective.

\begin{figure}
\centering
\epsfxsize=0.7\linewidth
    \includegraphics[width=8.8cm]{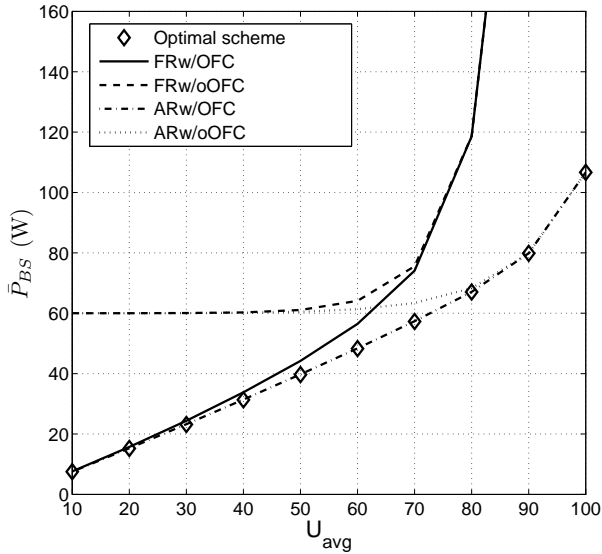}
\caption{Performance comparison with $P_c = 60$ W and $\bar{v} = 500$ Kbps} \label{fig:5_100_500_symmetric}
\end{figure}

To summarize, we draw the following key conclusions on the effects of different ESMs on the BS energy saving performance:
\begin{itemize}
\item BS on-off control is the most effective ESM when the network throughput is not high;
\item Cell range adaptation plays a more important role in BS energy saving when the network throughput is higher;
\item Finer-grained transmit power control at BS does not introduce significant benefit, i.e. constant power transmission at BS ``on'' mode is practically optimal.
\end{itemize}

\section{Conclusion}\label{sec:conclusion}
In this paper, under an OFDMA-based broadcast channel setup, we investigate optimal power and range adaptation polices with time-varying traffic to minimize the BS average power consumption subject to the throughput and QoS constraints. A new power scaling law that relates the (short-term) average transmit power at BS with the given cell range and MU density is derived, based on which we obtain the optimal power and range adaptation policy by solving a joint cell range adaptation and (long-term) power control problem. By exploiting the fact that energy saving at BS essentially comes from two major mechanisms, namely BS on-off power control and range adaptation, suboptimal schemes are proposed to achieve efficient performance-complexity tradeoffs. It is shown by simulation results that when the network throughput is modest, BS on-off power control is the most effective energy saving mechanism, while when the network throughput is higher, range adaptation becomes more effective. The results of this paper provide a preliminary unified framework for evaluating the performance of existing cell adaptation schemes such as BS's on-off switching and cell zooming, and for designing cell adaptation strategies for optimal energy saving.

{In this paper, we focus on the extreme case of a one-cell system for the purpose of obtaining useful insights, which needs to be extended to the more practical multi-cell scenario. It is thus interesting as well as important to investigate the optimal cell adaptation policy in a cooperative multi-cell setup by balancing between the cellular network energy consumption and its coverage performance by extending the mathematical framework developed in this paper.}

\appendices
\section{Proof of Theorem \ref{theorem:0}}\label{appendix:proof theorem 0}
First, $\mathbb{E}[P_i(r_i, n)]$ is computed based on (\ref{eq:pdf of location}) as follows, where $P_i(r_i, n)$ is given by (\ref{eq:single user transmission power}) with $N$ replaced by $n$.
\begin{align}\label{eq:single user expected transmission power}
\mathbb{E}[P_i(r_i, n)] = \frac{2\Gamma N_0W(2^{nC_2} - 1)}{KC_1(\alpha + 2)r_0^{\alpha}n}\left(R^{\alpha} + \frac{\alpha r_0^{\alpha + 2}}{2R^2}\right).
\end{align}
Since $\mathbb{E}[P_i(r_i, n)]$ is identical for all $i$'s, according to (\ref{eq:conditioned expected transmision power formular}), $\mathbb{E}[P_t|N]$ can be simply obtained through multiplying $\mathbb{E}[P_i(r_i, n)]$ by the number of MUs $n$, i.e.
\begin{align}\label{eq:conditioned expected transmision power}
\mathbb{E}[P_t|N] & = n\mathbb{E}[P_i(r_i, n)] \nonumber \\
& = \frac{2\Gamma N_0W(2^{nC_2} - 1)}{KC_1(\alpha + 2)r_0^{\alpha}}\left(R^{\alpha} + \frac{\alpha r_0^{\alpha + 2}}{2R^2}\right).
\end{align}
Averaging (\ref{eq:conditioned expected transmision power}) over the Poisson distribution of $N$, we finally obtain $\bar{P_t}$ as
\begin{align}
\bar{P_t}
& = \sum_{n = 0}^{\infty} \frac{2\Gamma N_0W(2^{nC_2} - 1)}{KC_1(\alpha + 2)r_0^{\alpha}}\left(R^{\alpha} + \frac{\alpha r_0^{\alpha + 2}}{2R^2}\right)\frac{\mu_N^n}{n!}e^{-\mu_N} \label{eq:total expected power through averaging over N} \\
& = D_1\left(R^{\alpha} + \frac{\alpha r_0^{\alpha + 2}}{2R^2}\right)\left(\sum_{n = 0}^{\infty}\frac{(\mu_N2^{C_2})^n}{n!}e^{-\mu_N} - 1\right) \\
& = D_1\left(R^{\alpha} + \frac{\alpha r_0^{\alpha + 2}}{2R^2}\right)\left(e^{D^{'}_1\pi\lambda R^2} - 1\right) \label{eq:before ignore r0}\\
& \approx D_1R^{\alpha}\left(e^{D^{'}_1\pi\lambda R^2} - 1\right)\label{eq:total power no approximaion}
\end{align}
where $D_1 = \frac{2\Gamma N_0W}{KC_1(\alpha + 2)r_0^{\alpha}}$ and $D^{'}_1 = 2^{\frac{\bar{v}}{W}} - 1$. Note that since cell radius $R$ is practically much larger than the reference distance $r_0$, we have ignored the term $\frac{\alpha r_0^{\alpha + 2}}{2R^2}$ in (\ref{eq:before ignore r0}).

It is worth noting that
\begin{align}
D^{'}_1  = (2^{\frac{\bar{v}}{W}}-1) = (2^{\frac{r_{se}}{\bar{N}}}-1)
\end{align}
where $r_{se}$ is the system spectrum efficiency in bps/Hz and $\bar{N}$ is the nominal number of supported users, both of which are pre-designed system parameters. In practice, $r_{se} = 2 \sim 6$ bps/Hz and $\bar{N}$ is a couple of hundreds and even thousands. Therefore, $\frac{r_{se}}{\bar{N}}$ is generally a very small number such that
\begin{align}
D^{'}_1 \approx \frac{\bar{v}}{W}\ln2.
\end{align}
Thus, (\ref{eq:total power no approximaion}) can be further simplified as
\begin{align}
\bar{P_t} \approx D_1R^{\alpha}\left(2^{D_2\pi\lambda R^2} - 1\right)
\end{align}
where $D_2 = \frac{\bar{v}}{W}$. Theorem \ref{theorem:0} is thus proved.

\section{Proof of Lemma \ref{lemma:0}}\label{appendix:proof lemma 0}
To prove Lemma \ref{lemma:0}, the following two facts are first verified:
\begin{enumerate}
\item For any $P_c$, which yields feasible (P0), there always exist some $\lambda$ such that $L_{\lambda}(x_1^{*}(\lambda), \mu) < 0$;
\item If $L_{\lambda}(x_1^{*}(\lambda_a), \mu) \leq 0$, then $L_{\lambda}(x_1^{*}(\lambda_b), \mu) < 0$ for all $\lambda_b > \lambda_a$.
\end{enumerate}

The first fact can be shown by contradiction as follows. Suppose that $L_{\lambda}(x_1^{*}(\lambda), \mu)$ is always non-negative, i.e.
\begin{align}
    L_{\lambda}(x_1^{*}(\lambda), \mu) \geq 0, ~~~~ \forall x > 0, \lambda \geq 0.
\end{align}
Then, according to (\ref{eq:original solution}) we have
\begin{align}
x^{*}(\lambda) = 0, ~~~~ \forall \lambda \geq 0
\end{align}
which violates the throughput constraint $\mathbb{E}_{\lambda}\left[U(x(\lambda), \lambda)\right] \geq U_{\text{avg}}$. The first fact is thus proved.

Next, we verify the second fact. According to the first fact, there always exists a $\lambda$ such that $L_{\lambda}(x_1^{*}(\lambda), \mu) < 0$. Therefore, without loss of generality, we can assume $L_{\lambda}(x_1^{*}(\lambda_a), \mu) \leq 0$, i.e.
\begin{align}
\min\limits_{x(\lambda_a) > 0} \bar{P}_{\text{BS}}(x(\lambda_a), \lambda_a) - \mu U(x(\lambda_a), \lambda_a) \leq 0.
\end{align}
Then there exists at least one $x_a(\lambda_a) > 0$ such that
\begin{align}
\bar{P}_{\text{BS}}(x_a(\lambda_a), \lambda_a) - \mu U(x_a(\lambda_a), \lambda_a) \leq 0
\end{align}
or equivalently,
\begin{align}
D_1x_a(\lambda_a)^{\frac{\alpha}{2}}\left(2^{D_2\pi\lambda_a x_a(\lambda_a)}-1\right) + P_{c} \leq \mu\pi\lambda_ax_a(\lambda_a).
\end{align}
For any given $\lambda_b > \lambda_a$, by letting $x_b(\lambda_b) = x_a(\lambda_a)\frac{\lambda_a}{\lambda_b}$, then
\begin{align}
&D_1x_b(\lambda_b)^{\frac{\alpha}{2}}\left(2^{D_2\pi\lambda_b x_b(\lambda_b)}-1\right) + P_{c} \\
= ~ &D_1x_b(\lambda_b)^{\frac{\alpha}{2}}\left(2^{D_2\pi\lambda_a x_a(\lambda_a)}-1\right) + P_{c} \\
< ~ &D_1x_a(\lambda_a)^{\frac{\alpha}{2}}\left(2^{D_2\pi\lambda_a x_a(\lambda_a)}-1\right) + P_{c} \\
\leq ~ &\mu\pi\lambda_ax_a(\lambda_a) = \mu\pi\lambda_bx_b(\lambda_b).
\end{align}
Thus for any $\lambda_b > \lambda_a$, we can always find an $x_b(\lambda_b)$ such that $\bar{P}_{\text{BS}}(x_b(\lambda_b), \lambda_b) - \mu U(x_b(\lambda_b), \lambda_b) < 0$, which implies $L_{\lambda}(x_1^{*}(\lambda_b), \mu) < 0$. The second fact is thus proved.

We are now ready to prove Lemma \ref{lemma:0}. The proof is by first showing the fact that $L_{\lambda}(x_1^{*}(\lambda), \mu)$ is positive for sufficiently small $\lambda$'s, and then combining this result with the two facts previously shown.

According to the first-order Taylor expansion, we have
\begin{align}
& D_1x(\lambda)^{\frac{\alpha}{2}}\left(2^{D_2\pi\lambda x(\lambda)}-1\right) + P_{c} \\
> ~ & (\ln 2)D_1D_2\pi\lambda x(\lambda)^{\frac{\alpha+2}{2}} + P_{c}, ~~ \forall x > 0.
\end{align}
Let $h(x(\lambda)) = (\ln 2)D_1D_2\pi\lambda x(\lambda)^{\frac{\alpha+2}{2}} + P_{c} - \mu\pi\lambda x(\lambda)$; then the minimum value of $h(x(\lambda))$ could be easily found by its first-order differentiation, given by
\begin{align}
h(x(\lambda))_{\text{min}} = P_{c} - x_{\text{min}}\lambda\mu\pi\frac{\alpha}{\alpha+2}
\end{align}
where $x_{\text{min}} = \left(\frac{2\mu}{(\alpha+2)(\ln 2)D_1D_2}\right)^{\frac{2}{\alpha}}$. It is easy to verify that if $\lambda < \frac{(\alpha+2)P_{c}}{\alpha\mu\pi x_{\text{min}}}$, $h(x(\lambda))_{\text{min}} > 0$. Since $L_{\lambda}(x(\lambda), \mu)$ is an upper bound of $h(x(\lambda))$, we have
\begin{align}
L_{\lambda}(x(\lambda), \mu) > 0, ~~ \forall x(\lambda) > 0 ~ \text{and} ~ \lambda < \frac{(\alpha+2)P_{c}}{\alpha\mu\pi x_{\text{min}}}
\end{align}
which implies that
\begin{align}\label{eq:small lambda positive}
L_{\lambda}(x_1^{*}(\lambda), \mu) > 0, ~~ \forall \lambda < \frac{(\alpha+2)P_{c}}{\alpha\mu\pi x_{\text{min}}}.
\end{align}
We thus show that $L_{\lambda}(x_1^{*}(\lambda_b), \mu)$ is positive for $\lambda$'s satisfying (\ref{eq:small lambda positive}). With the two facts given earlier, it follows that $L_{\lambda}(x_1^{*}(\lambda), \mu)$ cannot be positive for all $\lambda$'s and $L_{\lambda}(x_1^{*}(\lambda), \mu)$ will remain negative once it turns to be negative for the first time as $\lambda$ increases; thus, we conclude that there must exist a critical value for $\lambda$, i.e., $\lambda_1 > 0$ as given in Lemma \ref{lemma:0}. Lemma \ref{lemma:0} is thus proved.

\section{Proof of Lemma \ref{lemma:1}}\label{appendix:proof lemma 1}
Using the series expansion $2^{x} = \sum\limits_{k=0}^{\infty}\frac{(x(\ln2))^k}{k!}$, (\ref{eq:optimal condition for always on}) is expanded as
\begin{align}\label{eq:expanded optimal conditon for always on}
x_1^{*}(\lambda)^{\frac{\alpha}{2}}\sum_{k=1}^{\infty}\frac{(k+\frac{\alpha}{2})((\ln2)D_2\pi)^k(\lambda x_1^{*}(\lambda))^{k-1}}{k!}  = \frac{\mu\pi}{D_1}.
\end{align}
It can be verified that the left-hand-side (LHS) of (\ref{eq:expanded optimal conditon for always on}) is a strictly increasing function of both $\lambda$ and $x_1^{*}(\lambda)$. Thus, to maintain the equality in (\ref{eq:expanded optimal conditon for always on}), $x_1^{*}(\lambda)$ needs to be decreased when $\lambda$ increases and vice versa.

Since $U(x_1^{*}(\lambda), \lambda) = \pi\lambda x_1^{*}(\lambda)$, checking the monotonicity of $U(x_1^{*}(\lambda), \lambda)$ is equivalent to checking that of $\lambda x_1^{*}(\lambda)$. It is observed that if $\lambda$ increases, decreasing $x_1^{*}(\lambda)$ with $\lambda x_1^{*}(\lambda)$ being a constant will decrease the LHS of (\ref{eq:expanded optimal conditon for always on}) due to the term $x_1^{*}(\lambda)^{\frac{\alpha}{2}}$. Therefore, $\lambda x_1^{*}(\lambda)$ needs to be an increasing function of $\lambda$ and so does $U(x_1^{*}(\lambda), \lambda)$.

To prove the monotonicity of $\bar{P}_{\text{BS}}(x_1^{*}(\lambda), \lambda)$, we expand (\ref{eq:optimal condition for always on}) as
\begin{align}\label{eq:expanded for x1}
& \frac{\alpha}{2}D_1x_1^{*}(\lambda)^{\frac{\alpha-2}{2}}\left(2^{D_2\pi\lambda x_1^{*}(\lambda)}-1\right) \nonumber \\
+~ & (\ln 2)D_1D_2\pi\lambda x_1^{*}(\lambda)^{\frac{\alpha}{2}}2^{D_2\pi\lambda x_1^{*}(\lambda)} = \mu\pi\lambda
\end{align}
which can be rearranged as
\begin{align}
&D_1x_1^{*}(\lambda)^{\frac{\alpha}{2}}\left(2^{D_2\pi\lambda x_1^{*}(\lambda)}-1\right)\frac{\alpha}{2\lambda x_1^{*}(\lambda)} \nonumber \\
+~ & D_1x_1^{*}(\lambda)^{\frac{\alpha}{2}}\left(2^{D_2\pi\lambda x_1^{*}(\lambda)}-1\right)(\ln 2)\pi D_2 \nonumber \\
+~ & (\ln 2)D_1D_2\pi x_1^{*}(\lambda)^{\frac{\alpha}{2}} = \mu\pi
\end{align}
or equivalently,
\begin{align}\label{eq:prove P_bs monotonicity}
&\left(\bar{P}_{\text{BS}}(x_1^{*}(\lambda), \lambda) - P_c\right)\frac{\alpha}{2\lambda x_1^{*}(\lambda)} \nonumber \\
+~ & \left(\bar{P}_{\text{BS}}(x_1^{*}(\lambda), \lambda) - P_c\right)(\ln 2)\pi D_2 \nonumber \\
+~ & (\ln 2)D_1D_2\pi x_1^{*}(\lambda)^{\frac{\alpha}{2}} = \mu\pi.
\end{align}
Suppose that $x_1^{*}(\lambda_1)$ and $x_1^{*}(\lambda_2)$ are the two roots of (\ref{eq:optimal condition for always on}) when $\lambda = \lambda_1$ and $\lambda = \lambda_2$, respectively, where $\lambda_2 > \lambda_1$. Based on the monotonicity of $x_1^{*}(\lambda)$ and $U(x_1^{*}(\lambda), \lambda)$ proved above, we have
\begin{align}
    x_1^{*}(\lambda_1)\lambda_1 & < x_1^{*}(\lambda_2)\lambda_2, \\
    x_1^{*}(\lambda_1)^{\frac{\alpha}{2}} & > x_1^{*}(\lambda_2)^{\frac{\alpha}{2}}.
\end{align}
Due to the equality in (\ref{eq:prove P_bs monotonicity}) for all $\lambda > 0$, we have
\begin{align}
\bar{P}_{\text{BS}}(x_1^{*}(\lambda_1), \lambda_1) < \bar{P}_{\text{BS}}(x_1^{*}(\lambda_2), \lambda_2), \forall \lambda_2 > \lambda_1.
\end{align}
Lemma \ref{lemma:1} is thus proved.

\section{Proof of Theorem \ref{theorem:1}}\label{appendix:proof theorem 1}
First, we consider the case of $\lambda_2 \geq \lambda_1$, in which three subcases are addressed as follows:
    \begin{enumerate}
        \item If $\lambda \leq \lambda_1$, according to the definition of $\lambda_1$ given in Lemma \ref{lemma:0}, $L_{\lambda}(x_1^{*}(\lambda), \mu) \geq 0$ for $\lambda \leq \lambda_1$, which corresponds to the third condition in (\ref{eq:original solution}). Therefore, we have
            \begin{align}
            x^{*}(\lambda) = 0. \nonumber
            \end{align}
        \item If $\lambda_1 < \lambda \leq \lambda_2$, we have $L_{\lambda}(x_1^{*}(\lambda), \mu) < 0$. Since $\bar{P}_{\text{BS}}(x_1^{*}(\lambda_2), \lambda_2) = P_{\text{max}}$ and $\bar{P}_{\text{BS}}(x_1^{*}(\lambda), \lambda_2)$ increases with $\lambda$ from Lemma \ref{lemma:1}, it can be easily verified that $\bar{P}_{\text{BS}}(x_1^{*}(\lambda), \lambda_1) < P_{\text{max}}$ for the assumed range of $\lambda$, which is in accordance with the first condition in (\ref{eq:original solution}). Therefore, we have
            \begin{align}
            x^{*}(\lambda) = x_1^{*}(\lambda). \nonumber
            \end{align}
        \item Otherwise, if $\lambda > \lambda_2 \geq \lambda_1$, similar to the previous subcase, we know that $\bar{P}_{\text{BS}}(x_1^{*}(\lambda), \lambda_1) > P_{\text{max}}$. Next, we need to check the sign of $L_{\lambda}(x_2^{*}(\lambda), \mu) = P_{\text{max}} - \mu\pi\lambda x_2^{*}(\lambda)$. Note that $L_{\lambda}(x_2^{*}(\lambda_2), \mu) = L_{\lambda}(x_1^{*}(\lambda_2), \mu)$, which is non-positive due to $\lambda_2 \geq \lambda_1$. Since $U(x_2^{*}(\lambda), \lambda)$ strictly increases with $\lambda$, $L_{\lambda}(x_2^{*}(\lambda), \mu)$ is thus a strictly decreasing function of $\lambda$. Therefore $L_{\lambda}(x_2^{*}(\lambda), \mu) < 0$ for $\lambda > \lambda_2$, which implies
            \begin{align}
            x^{*}(\lambda) = x_2^{*}(\lambda). \nonumber
            \end{align}
    \end{enumerate}
Second, consider the case of $\lambda_2 < \lambda_1$. It is first verified that $\lambda_3 > \lambda_1 > \lambda_2$ in this case as follows: since $x_1^{*}(\lambda_1)$ minimizes $L_{\lambda}(x(\lambda), \mu)$ when $\lambda = \lambda_1$ to attain a zero value, and $L_{\lambda}(x(\lambda), \mu)$ is strictly convex in $x(\lambda)$, it follows that $L_{\lambda}(x_2^{*}(\lambda_1), \mu) > 0$. Since $L_{\lambda}(x_2^{*}(\lambda_3), \mu) = 0$ and $L_{\lambda}(x_2^{*}(\lambda), \mu)$ is a strictly decreasing function of $\lambda$, we conclude that $\lambda_3 > \lambda_1$. Next, we consider the following three subcases:
    \begin{enumerate}
        \item If $\lambda \leq \lambda_1$, according to Lemma \ref{lemma:0}, it is easy to verify that $L_{\lambda}(x_2^{*}(\lambda), \mu) > L_{\lambda}(x_1^{*}(\lambda), \mu) \geq 0$. Therefore, we have
        \begin{align}
        x^{*}(\lambda) = 0. \nonumber
        \end{align}
        \item If $\lambda_1 < \lambda \leq \lambda_3$, we have $\bar{P}_{\text{BS}}(x_1^{*}(\lambda), \lambda) > P_{\text{max}}$ and $L_{\lambda}(x_2^{*}(\lambda), \mu) \geq 0$, which implies
        \begin{align}
        x^{*}(\lambda) = 0. \nonumber
        \end{align}
        \item Otherwise, if $\lambda > \lambda_3$, we have $\bar{P}_{\text{BS}}(x_1^{*}(\lambda), \lambda) > P_{\text{max}}$ and $L_{\lambda}(x_2^{*}(\lambda), \mu) < 0$, which is in accordance with the second condition in (\ref{eq:original solution}). Therefore, we have
            \begin{align}
            x^{*}(\lambda) = x_2^{*}(\lambda). \nonumber
            \end{align}
    \end{enumerate}
Combining the above two cases, Theorem \ref{theorem:1} is thus proved.

\section{Proof of Lemma \ref{lemma:3}}\label{appendix:proof lemma 3}
From (\ref{eq:power radium relationship}) and (\ref{eq:optimal condition for always on}), we obtain the following equation
\begin{align}\label{eq:equation for x2}
D_1x_2^{*}(\lambda)^{\frac{\alpha}{2}}\left(2^{D_2\pi\lambda x_2^{*}(\lambda)} - 1\right) = P_{\text{max}} - P_c.
\end{align}
With the HSE assumption of $D_2\pi\lambda x_2^{*}(\lambda) \gg 1$, (\ref{eq:equation for x2}) is simplified as
\begin{align}\label{eq:approximated equation for x2}
D_1x_2^{*}(\lambda)^{\frac{\alpha}{2}}2^{D_2\pi\lambda x_2^{*}(\lambda)} = P_{\text{max}} - P_c
\end{align}
which can be rearranged as
\begin{align}\label{eq:rearrange approximated equation for x2}
2^{-\frac{2D_2\pi\lambda}{\alpha}x_2^{*}(\lambda)} = \left(\frac{D_1}{P_{\text{max}}- P_c}\right)^{\frac{2}{\alpha}}x_2^{*}(\lambda).
\end{align}
By utilizing
\begin{align}\label{eq:w formular}
p^{ax+b} = cx+d \Rightarrow x = - \frac{\mathcal{W}\left(-\frac{a\ln p}{c}p^{b - \frac{ad}{c}}\right)}{a\ln p} - \frac{d}{c}
\end{align}
with $p > 0$, $a, c \neq 0$, it is easy to verify that $a = -\frac{2D_2\pi\lambda}{\alpha}$, $b=0$, $c = \left(\frac{D_1}{P_{\text{max}}- P_c}\right)^{\frac{2}{\alpha}}$, $d=0$ and $p=2$ in (\ref{eq:rearrange approximated equation for x2}). Thus, $x_2^{*}(\lambda)$ is given by
\begin{align}
x_2^{*}(\lambda) = \frac{\alpha}{2D_3\pi\lambda}\mathcal{W}\left(\frac{2D_3\pi\lambda}{\alpha}\left(\frac{P_{\text{max}}- P_c}{D_1}\right)^{\frac{2}{\alpha}}\right).
\end{align}

We then proceed to derive the expression of $x_1^{*}(\lambda)$. Note that $x_1^{*}(\lambda)$ is the root of equation (\ref{eq:expanded for x1}), which can be expressed as
\begin{align}\label{eq:approximated equation for x1}
x_1^{*}(\lambda)^{\frac{\alpha-2}{2}}2^{D_2\pi\lambda x_1^{*}(\lambda)}\left[\frac{\alpha}{2} + (\ln2)D_2\pi\lambda x_1^{*}(\lambda)\right] = \frac{\mu\pi\lambda}{D_1}
\end{align}
by applying the HSE assumption of $D_2\pi\lambda x_1^{*}(\lambda) \gg 1$. Furthermore, it is observed that (\ref{eq:approximated equation for x1}) can be simplified as
\begin{align}\label{eq:further approximated equation for x1}
(\ln2)D_1D_2x_1^{*}(\lambda)^{\frac{\alpha}{2}}2^{D_2\pi\lambda x_1^{*}(\lambda)} = \mu
\end{align}
due to the fact that $(\ln2)D_2\pi\lambda x_1^{*}(\lambda) \gg \frac{\alpha}{2}$, where $\alpha = 2 \sim 6$ in practice. Similar to the case for obtaining $x_2^{*}(\lambda)$, $x_1^{*}(\lambda)$ can be solved from (\ref{eq:further approximated equation for x1}) and given by
\begin{align}
x_1^{*}(\lambda) = \frac{\alpha}{2D_3\pi\lambda}\mathcal{W}\left(\frac{2D_3\pi\lambda}{\alpha}\left(\frac{\mu}{D_1D_3}\right)^{\frac{2}{\alpha}}\right).
\end{align}
Lemma \ref{lemma:3} is thus proved.

\end{document}